\documentclass[journal]{IEEEtran}

\usepackage{empheq}  
\usepackage{xcolor}
\definecolor{jxempheq}{RGB}{85, 239, 196}

\usepackage{amssymb}  
\usepackage{amsthm, amsmath}
\usepackage{mathtools}
\usepackage{extarrows}  
\usepackage{dsfont}

\newtheorem{remark}{Remark}
\newtheorem{theorem}{Theorem}

\newtheorem{corollary}{Corollary}[theorem]

\DeclarePairedDelimiter\abs{|}{|}  
\DeclarePairedDelimiter\norm{\|}{\|}
\providecommand\given{\:\vert\:}
\DeclarePairedDelimiterXPP\set[1]{}{\{}{\}}{}{\renewcommand\given{\nonscript\:\delimsize\vert\nonscript\:\mathopen{}}#1}

\DeclareMathOperator{\T}{T}

\DeclareMathOperator{\dif}{d}

\DeclareMathOperator{\diag}{diag}

\usepackage[caption=false,font=footnotesize,labelfont=rm,textfont=rm]{subfig}
\usepackage{url}
\usepackage[hidelinks,breaklinks]{hyperref}
\hypersetup{
    colorlinks=true,
    linkcolor=blue,
    filecolor=black,
    urlcolor=black,
    citecolor=blue,
}
\usepackage{cite}
\usepackage{array, arydshln, xtab}
\usepackage{multirow, booktabs, threeparttable}

\usepackage[ruled,vlined]{algorithm2e}
\SetAlgoProcName{Alg.}{Alg.}

\SetKwInOut{Input}{Input}
\SetKwInOut{Output}{Output}
\SetKwInOut{Initialization}{Initial}

\newcommand{\tabref}[1]{Table~\ref{#1}}
\newcommand{\figref}[1]{Fig.~\ref{#1}}
\renewcommand{\eqref}[1]{\textcolor{blue}{(\ref{#1})}}
\makeatletter
\def\@cite#1#2{{\color{blue}[{#1\if@tempswa , #2\fi}]}}
\makeatother

\usepackage{soul}  
\soulregister\cite7
\soulregister\eqref7
\soulregister\ref7
\soulregister\section7
\soulregister\subsection7

\usepackage{tikz}
\usetikzlibrary{arrows,shapes,positioning,chains,matrix,intersections,calc}
\usetikzlibrary{arrows.meta}
\usetikzlibrary{decorations.markings, decorations.pathreplacing}
\tikzstyle{block} = [draw, rectangle, minimum height=2em, minimum width=3.75em]
\tikzstyle{cblock} = [circle, draw, line width=2pt, inner sep=0pt, minimum size=0.75cm, font=\fontsize{9pt}{\baselineskip}\selectfont]
\tikzstyle{cblock2} = [circle, draw, line width=2pt, inner sep=1pt, minimum size=1.5cm, font=\fontsize{10pt}{\baselineskip}\selectfont]
\tikzstyle{triblk} = [draw, isosceles triangle, isosceles triangle apex angle=55, align=center,shape border rotate=0]
\tikzstyle{sum} = [draw, circle, minimum size=0.3cm]
\tikzstyle{input} = [coordinate]
\tikzstyle{output} = [coordinate]
\usepackage{pgfplots}
\pgfplotsset{compat=1.17}

\definecolor{jxgreen}{RGB}{120,224,143}
\definecolor{jxgreen2}{RGB}{184,233,148}
\definecolor{jxblue}{RGB}{96,163,188}
\definecolor{jxblue2}{RGB}{130,204,221}
\definecolor{jxorange}{RGB}{229,80,57}
\definecolor{jxorange2}{RGB}{248,194,145}
\definecolor{jxedge}{RGB}{106,137,204}

\newcommand{\jxrevise}[1]{\textcolor{black}{#1}}

\begin{document}

\title{System Strength-Constrained Scheduling with Switchable Grid-Forming and Grid-Following Generation Resources}

\author{Jiaxin~Wang,~\IEEEmembership{Graduate~Student~Member,~IEEE,}
    Fei~Teng,~\IEEEmembership{Senior~Member,~IEEE,}
    Huanhai~Xin,~\IEEEmembership{Senior~Member,~IEEE,}
    Jing~Dai,~\IEEEmembership{Member,~IEEE,}
    Tomislav Capuder,~\IEEEmembership{Member,~IEEE,}
    and Ning~Zhang,~\IEEEmembership{Senior~Member,~IEEE}
    {
        \scriptsize
        \thanks{
            Jiaxin~Wang and Ning~Zhang are with the State Key Laboratory of Power System Operation and Control, Department of Electrical Engineering, Tsinghua University, Beijing 100084, China.
            (email: \href{mailto:wjx22@mails.tsinghua.edu.cn}{wjx22@mails.tsinghua.edu.cn}; \href{mailto:ningzhang@tsinghua.edu.cn}{ningzhang@tsinghua.edu.cn}).

            Fei~Teng is with the Department of Electrical and Electronic Engineering, Imperial College London, London SW7 2AZ, U.K.
            (email: \href{mailto:f.teng@imperial.ac.uk}{f.teng@imperial.ac.uk}).

            Huanhai~Xin is with the College of Electrical Engineering, Zhejiang University, Hangzhou 310027, China.
            (email: \href{mailto:xinhh@zju.edu.cn}{xinhh@zju.edu.cn}).

            Jing~Dai is with the Think Tank Research Center, Tsinghua University, Beijing 100084, China.
            (email: \href{mailto:jingdai@tsinghua.edu.cn}{jingdai@tsinghua.edu.cn}).

            Tomislav~Capuder is with the Department of Energy and Power Systems, University of Zagreb, Zagreb 10000, Croatia. (email: \href{mailto:tomislav.capuder@fer.unizg.hr}{tomislav.capuder@fer.unizg.hr}).
        }
    }
}
\maketitle

\begin{abstract}
    Inverter-based resources (IBRs) are increasingly dominating modern power systems, posing significant challenges to cost-effectively maintain system strength for stability.
    At the same time, the operating behaviors of IBRs are software-defined, including both their steady-state power outputs and control modes, e.g. grid-forming (GFM) and grid-following (GFL).
    Such flexibility has not been fully explored to efficiently operate future power systems.
    This paper develops a novel framework that simultaneously optimizes IBR operating behaviors and ensures adequate system strength.
    A comprehensive solution is provided to integrate system strength constraints into scheduling models, despite their inherent strong non-convexities.
    We derive a rigorous linear-matrix-inequality (LMI) reformulation of the system strength constraint, effectively addressing non-explicit formulations and dimension variation issues caused by GFM/GFL mode switching of IBRs.
    Then, we equivalently convert the original non-convex implicit system strength-constrained scheduling problem into an explicit mixed-integer semi-definite programming (MISDP) problem by incorporating the reformulated system strength constraint along with other operational constraints.
    We further provide a Rayleigh Cut method, which is compatible with standard mixed-integer linear programming (MILP) solvers, to solve this system strength-constrained scheduling problem.
    Case studies on a modified IEEE 118-bus system and a practical Jiangsu power system demonstrate the performance of the proposed methods.
\end{abstract}


\section*{Nomenclature}

\subsection{Abbreviations}

\begin{IEEEdescription}[\IEEEusemathlabelsep\IEEEsetlabelwidth{$V_1,V_2,V_3$}]
    \item[gOSCR] {Generalized Operational Short-Circuit Ratio}
    \item[GFM] {Grid-Forming}
    \item[GFL] {Grid-Following}
    \item[IBR] {Inverter-Based Resource}
    \item[MIP] {Mixed-Integer Programming}
    \item[MILP] {Mixed-Integer Linear Programming}
    \item[MISDP] {Mixed-Integer Semi-Definite Programming}
    \item[PV] {Photovoltaic}
    \item[LMI] {Linear Matrix Inequality}
    \item[SDP] {Semi-Definite Programming}
    \item[SG] {Synchronous Generator}
    \item[SOC] {State of Charge}
\end{IEEEdescription}

\subsection{Sets}

\begin{IEEEdescription}[\IEEEusemathlabelsep\IEEEsetlabelwidth{$V_1,V_2$}]
    \item[\(\mathcal{E}\)] {Buses with energy storage systems}
    \item[\(\mathcal{R}\)] {Buses with renewable energy sources}
    \item[\(\mathcal{G}\)] {Buses with synchronous generators}
    \item[\(\mathcal{L}\)] {Remaining Buses}
    \item[\(\mathcal{N}\)] {All buses \(\mathcal{N}=\mathcal{E}\cup\mathcal{R}\cup\mathcal{G}\cup\mathcal{L}\)}
    \item[\(\mathcal{I}\)] {IBR buses \(\mathcal{I}=\mathcal{E}\cup\mathcal{R}\)}
    \item[\(\mathcal{J}\)] {Other buses \(\mathcal{J}=\mathcal{G}\cup\mathcal{L}=\mathcal{N}\setminus\mathcal{I}\)}
    \item[\(\mathcal{T}\)] {Time horizons}
    \item[\(\mathcal{E}^{\mathrm{GFM}}_t\)] {Energy storage operating in GFM at time \(t\in\mathcal{T}\)}
    \item[\(\mathcal{E}^{\mathrm{GFL}}_t\)] {Energy storage operating in GFL at time \(t\in\mathcal{T}\)}
    \item[\(\mathcal{R}^{\mathrm{GFM}}_t\)] {Renewable operating in GFM at time \(t\in\mathcal{T}\)}
    \item[\(\mathcal{R}^{\mathrm{GFL}}_t\)] {Renewable operating in GFL at time \(t\in\mathcal{T}\)}
    \item[\(\mathcal{I}_t\)] {GFL buses \(\mathcal{I}_t = \mathcal{E}^{\mathrm{GFL}}_t\cup \mathcal{R}^{\mathrm{GFL}}_t\) at time \(t\in\mathcal{T}\)}
    \item[\(\mathcal{J}_t\)] {Non-GFL buses \(\mathcal{J}_t = \mathcal{N}\setminus\mathcal{I}_t\) at time \(t\in\mathcal{T}\)}
\end{IEEEdescription}

\subsection{Variables}

\begin{IEEEdescription}[\IEEEusemathlabelsep\IEEEsetlabelwidth{$V_1,V_2$}]
    \item[\(\gamma\)] {gOSCR}
    \item[\(\gamma_0\)] {gOSCR threshold for system strength adequacy}
    \item[\(x_{i,t}\)] {Binary variable indicating the operating mode of generation resource \(i\in\mathcal{I}\) at time \(t\in\mathcal{T}\), where \(x_{i,t}=1\) for GFM and \(x_{i,t}=0\) for GFL}
    \item[\(B^{\mathrm{pf}}\)] {Network admittance matrix used in the power flow model}
    \item[\(B^{\mathrm{sys}}\)] {System admittance matrix accounting for the internal admittances of online synchronous generators and GFM IBRs}
    \item[\(B\)] {Schur complement of the system admittance matrix reduced to the GFL buses}
    \item[\(\hat{B}\)] {Schur complement of the system admittance matrix reduced to the IBR buses}
    \item[\(P\)] {Diagonal matrix of GFL power injections}
    \item[\(\hat{P}\)] {Diagonal matrix of GFL power injections extended to all IBR buses, with zeros padded at non-GFL buses}
\end{IEEEdescription}

\section{Introduction}

\IEEEPARstart{T}{he} rapid growth of inverter-based resources (IBRs) is fundamentally reshaping modern power systems \cite{nikos2021definition,qingchun2020impact}.
Compared with traditional power systems, inverter-dominated systems replace large shares of synchronous generators with renewable energy sources and energy storage systems interfaced through power electronics \cite{yunjie2023power}.
While this transformation is essential for enable decarbonization, it may also result in insufficient system strength, potentially leading to IBR connection failures and severe stability issues such as system oscillations \cite{guoxuan2025many}.
However, conventional scheduling method determines operating decisions primarily based on steady-state constraints, with stability verified afterward.
It becomes inadequate for IBRs-dominated systems where system strength is highly sensitive to operational decisions of IBRs \cite{kehao2026quantitative} and growing difficult to rely on post-scheduling stability checks and manual adjustments to maintain both economic efficiency and operational security \cite{ning2023data}.
Therefore, it is essential to incorporate system strength constraints directly into the scheduling problem and develop system strength-constrained scheduling frameworks that simultaneously optimize IBR operating behaviors and ensure adequate system strength \cite{Kamwa2026future}.

System strength is strongly coupled with IBR operating behaviors \cite{linbin2026system}.
Traditionally, system strength was almost irrelevant to operating points as it was mainly determined by the network structure and the on/off status of synchronous generators \cite{krishayya1997ieee}.
However, in IBR-dominated power systems, complex interactive influences between system strength and IBR operating conditions arise.
In GFL mode, the steady-state power outputs of IBRs directly affect system strength as the dynamic performance is closely related to their operating points through phase-locked loop and outer loops \cite{yunjie2021impedance}.
In GFM mode, while providing support to system strength, IBRs must reserve part of their capacity in steady state operation to maintain approximately constant voltage during dynamics, which reduces the ability to support power balance and incurs higher operational costs \cite{yitong2022revisiting}.
\figref{fig:gfm-penetration} illustrates a possible distribution of system strength across different GFM/GFL ratios and operating points, suggesting that system strength can vary significantly with different IBR operating behaviors.

\begin{figure}[!t]
    \centering
    \includegraphics[width=0.86\linewidth]{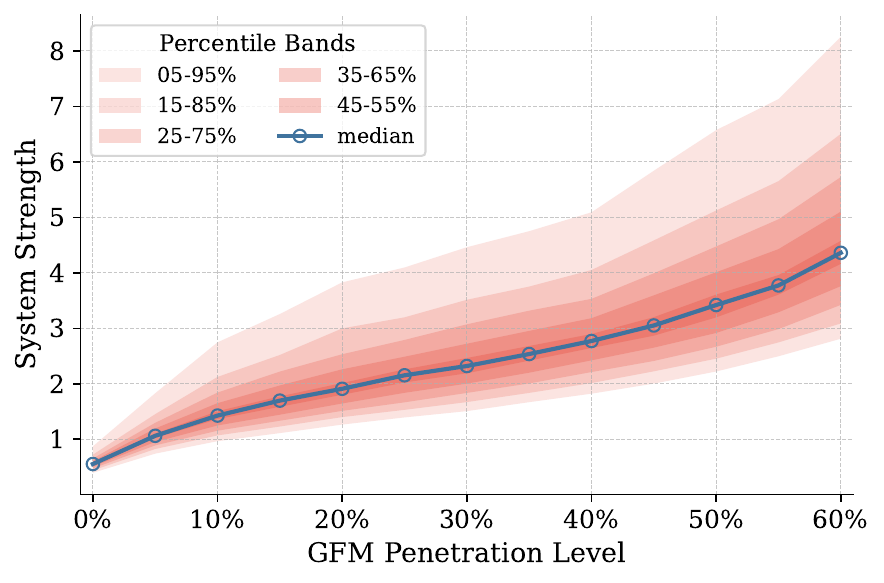}
    \caption{System strength may exhibit significant variation across different GFM/GFL ratios and operating points. Here, the system strength is measured by the gOSCR metric, and the operating points are generated by randomly sampling the power outputs of IBRs.}
    \label{fig:gfm-penetration}
\end{figure}

At the same time, the behaviors of IBRs, including both their steady-state power outputs and control modes, are software-defined and thus exhibit substantial flexibility to be optimized for enhancing system strength and improving economic efficiency \cite{liwei2023hierarchical}.
In particular, within their available generation limits, IBRs can be dispatched over a wide operating range with high efficiency, as their output adjustments are realized through fast power electronic control rather than mechanical actuation \cite{mahdi2025adaptive}.
\jxrevise{Beyond power dispatch, recent studies have shown that IBRs can also be equipped with device-level control schemes that enable flexible switching between GFL and GFM modes.
    For example, switchable GFL/GFM control architectures have been developed to support mode transition under different grid conditions \cite{huazhao2025novel}, and seamless switching strategies have also been proposed to reduce transient impacts during the transition between GFL and GFM operation \cite{niu2026seamless}.
    More importantly from a scheduling perspective, such mode-switching capability allows IBRs to operate in either GFL or GFM mode, provided that the required power reserve for GFM operation is satisfied.
    Such flexibility fundamentally transforms the paradigm of power system scheduling, offering new opportunities to simultaneously achieve system security and economic efficiency \cite{chu2023stabilityconstrainedoptimizationhigh}.}
Unfortunately, it also introduces great complexity to the scheduling problem, as the coupling between system strength and IBR operating behaviors must be explicitly captured and optimized within the scheduling framework.

To capture the complicated interactions between system strength and IBR operating behaviors, several metrics, including the grid strength impedance metric (GSIM) \cite{callum2024grid}, impedance margin ratio (IMR) \cite{yue2024impedance}, and generalized short-circuit ratio (gSCR) \cite{fuyilong2024assessing,yuanhui2025placing}, have been proposed to quantify system strength, thereby replacing traditional heuristic metrics \cite{cigre2008b4}.
While these metrics are rigorously derived and can effectively capture the interactions, the relationship between system strength and IBR operating behaviors is still non-explicit and only implicitly reflected through the whole-system impedance model, which may not be straightforward to be incorporated into scheduling problems.
To bridge this gap, generalized operational short-circuit ratio (gOSCR) has been developed to extract key scheduling-relevant variables such as power outputs and GFM/GFL mode selections of IBRs in a more explicit form \cite{chenxi2024generalized}, which is promising to be incorporated into system strength-constrained scheduling models.

Several studies have investigated the integration of system strength constraints into power system optimization problems.
Early works \cite{yongkyu2022evaluation,zhongda2023voltage,yongkyu2025strength} developed port equivalence, quadratic regression, and linearization techniques to embed system strength constraints into the unit commitment problem.
While these methods enable tractable mixed-integer linear programming (MILP) formulations, they rely on simplified traditional metrics that cannot fully capture the complex interactions with IBR operating behaviors.
To overcome these limitations, recent studies have begun to incorporate mathematically more rigorous system strength metrics into power system optimization models.
One representative direction is to approximate control-related system strength constraints using data-driven surrogate models, which can then be reformulated as mixed-integer linear constraints and integrated into scheduling problems \cite{guoxuan2025control,juelin2022explicit}.
This type of method provides a practical way to embed otherwise complex stability-related constraints into tractable MILP formulations.
However, its accuracy inherently depends on the representativeness of the sampled operating points and the quality of the fitted surrogate model.
Therefore, it cannot provide a strict theoretical guarantee that scheduling decisions feasible for the surrogate constraints always satisfy the original system strength requirements.
Another direction seeks analytically derived reformulations of system strength constraints.
For example, linear-matrix-inequality-based (LMI-based) formulations of system strength constraints have been developed to directly incorporate stability requirements into optimization problems without data-fitting approximation errors \cite{huanhai2025many,jiaxin2025synchronous}.
Nevertheless, these LMI-based formulations usually simplify the operational modeling of IBRs when embedded into scheduling frameworks.
In particular, the complex coupling between system strength and IBR flexibility, including GFL/GFM mode switching, energy storage charging and discharging, and the power headroom required for GFM operation, has not been fully captured.
As a result, existing studies either rely on tractable approximations without strict feasibility guarantees or adopt rigorous strength representations under simplified operational modeling, leaving a gap for a computationally tractable scheduling framework that can fully exploit IBR operational flexibility while ensuring adequate system strength.

In addition to modeling these complex constraints, another key challenge lies in solving the resulting mixed-integer scheduling problem efficiently.
Cutting-plane-based algorithms have become mature tools for solving unit commitment and other power system scheduling problems with complex operational constraints.
A classical example is Benders-decomposition-based security-constrained unit commitment, where the unit commitment master problem is iteratively strengthened by feasibility and optimality cuts generated from network-security subproblems \cite{fu2013modeling}.
Recent studies have further adapted cutting-plane ideas to stability-related scheduling problems.
For example, column-and-constraint generation has been used to solve small-signal-stability-constrained operation of inverter-dominated microgrids \cite{wang2024stability}, while Benders decomposition has been applied to learning-assisted transient-stability-constrained unit commitment \cite{wu2024transient}.
The key to an effective cutting-plane framework is not the generic iterative structure itself, but the automatic generation of effective, problem-specific cuts.
However, for the system strength-constrained scheduling problem considered in this paper, the constraints are tightly coupled with operating decisions of IBRs, and the original feasible region is highly non-convex and non-linear, making it difficult to generate effective cuts.

To address the existing gaps, this paper proposes a system strength-constrained scheduling model that optimizes IBR operating behaviors while ensuring adequate system strength through the gOSCR metric.
First, we formulate a comprehensive system strength-constrained scheduling model that captures the coupling between system strength and various scheduling decisions, including the on/off status of synchronous generators, GFM/GFL mode switching of IBRs, and their power outputs.
Then, we propose and rigorously prove an equivalent LMI reformulation of the system strength constraint, converting the original difficult scheduling problem into an explicit mixed-integer semidefinite programming (MISDP) problem.
To the best of our knowledge, this is the first work to exactly transform the system strength constraint with complex coupling between operational decisions and system strength into an MISDP formulation without approximation error.
Finally, we provide a Rayleigh Cut solution method to solve the problem, guaranteeing security and optimality (within a prescribed tolerance) while being compatible with standard MILP solvers.

The main contributions of this paper are summarized:
\begin{itemize}
    \item A comprehensive system strength-constrained scheduling model is formulated that captures the coupling between system strength and various scheduling decisions. It can be also applied to analyze how the system strength constraint impacts the operational behavior of IBRs, including their power outputs and GFL/GFM mode selections.
    \item A MISDP reformulation is derived for the proposed system strength-constrained scheduling model without approximation errors through a series of rigorous proofs, rendering the originally intractable system strength constraints amenable to optimization-based scheduling.
    \item A Rayleigh Cut solution method is provided to solve the resulting system strength-constrained scheduling problem, which guarantees security and optimality while being compatible with standard MILP solvers.
\end{itemize}

The remainder of this paper is organized as follows.
Section~\ref{sec:motivations} formulates the overall system strength-constrained scheduling model in abstract form, and provides a complete view of the challenges involved and the solution framework.
Section~\ref{sec:formulation} proposes a detailed formulation of the IBR operating constraints and discusses the strong coupling between them and the system strength constraint.
Section~\ref{sec:reformulation} reformulates the system strength constraint into an equivalent LMI form and derives the MISDP reformulation of the scheduling model.
A Rayleigh Cut solution method is provided to solve the obtained model.
Section~\ref{sec:cases} conducts numerical experiments on a modified IEEE 118-bus system to demonstrate the performance of the proposed methods.
Section~\ref{sec:conclusions} concludes this paper.

\section{Problem Statement and the Proposed Solution\label{sec:motivations}}

\begin{figure*}[hbtp]
    \begin{equation}
        \label{eq:B-ori}
        B^{\mathrm{sys}}=\begin{tikzpicture}[baseline=(m.center)]
            \small
            \matrix (m) [
            matrix of math nodes,
            left delimiter={[},
            right delimiter={]},
            row sep=3.5pt,
            column sep=5pt
            ]
            {
            B^{\mathrm{sys}}_{\mathcal{E}^{\mathrm{GFL}}_t\mathcal{E}^{\mathrm{GFL}}_t} & \ast & \ast & \ast & \ast & \ast \\
            B^{\mathrm{sys}}_{\mathcal{R}^{\mathrm{GFL}}_t\mathcal{E}^{\mathrm{GFL}}_t} & B^{\mathrm{sys}}_{\mathcal{R}^{\mathrm{GFL}}_t\mathcal{R}^{\mathrm{GFL}}_t} & \ast & \ast & \ast & \ast \\
            B^{\mathrm{sys}}_{\mathcal{E}^{\mathrm{GFM}}_t\mathcal{E}^{\mathrm{GFL}}_t} & B^{\mathrm{sys}}_{\mathcal{E}^{\mathrm{GFM}}_t\mathcal{R}^{\mathrm{GFL}}_t} & B^{\mathrm{sys}}_{\mathcal{E}^{\mathrm{GFM}}_t\mathcal{E}^{\mathrm{GFM}}_t} & \ast & \ast & \ast \\
            B^{\mathrm{sys}}_{\mathcal{R}^{\mathrm{GFM}}_t\mathcal{E}^{\mathrm{GFL}}_t} & B^{\mathrm{sys}}_{\mathcal{R}^{\mathrm{GFM}}_t\mathcal{R}^{\mathrm{GFL}}_t} & B^{\mathrm{sys}}_{\mathcal{R}^{\mathrm{GFM}}_t\mathcal{E}^{\mathrm{GFM}}_t} & B^{\mathrm{sys}}_{\mathcal{R}^{\mathrm{GFM}}_t\mathcal{R}^{\mathrm{GFM}}_t} & \ast & \ast \\
            B^{\mathrm{sys}}_{\mathcal{G}\mathcal{E}^{\mathrm{GFL}}_t} & B^{\mathrm{sys}}_{\mathcal{G}\mathcal{R}^{\mathrm{GFL}}_t} & B^{\mathrm{sys}}_{\mathcal{G}\mathcal{E}^{\mathrm{GFM}}_t} & B^{\mathrm{sys}}_{\mathcal{G}\mathcal{R}^{\mathrm{GFM}}_t} & B^{\mathrm{sys}}_{\mathcal{G}\mathcal{G}} & \ast \\
            B^{\mathrm{sys}}_{\mathcal{L}\mathcal{E}^{\mathrm{GFL}}_t} & B^{\mathrm{sys}}_{\mathcal{L}\mathcal{R}^{\mathrm{GFL}}_t} & B^{\mathrm{sys}}_{\mathcal{L}\mathcal{E}^{\mathrm{GFM}}_t} & B^{\mathrm{sys}}_{\mathcal{L}\mathcal{R}^{\mathrm{GFM}}_t} & B^{\mathrm{sys}}_{\mathcal{L}\mathcal{G}} & B^{\mathrm{sys}}_{\mathcal{L}\mathcal{L}} \\
            };

            \draw ($(m-2-1.south) + (-2.7em, -2pt)$) -- ++(10.34, 0);
            \draw ($(m-1-2.east) + (2.6em, 0.8em)$) -- ++(0, -4.55);

            \node (m16) at (m-1-6) {};
            \node (m26) at (m-2-6) {};
            \node (m36) at (m-3-6) {};
            \node (m46) at (m-4-6) {};
            \node (m56) at (m-5-6) {};
            \node (m66) at (m-6-6) {};
            \node[right=18pt of m16] (R1) {\color{jxorange}$\mathcal{E}_t^{\mathrm{GFL}}$};
            \node (R2) at ($(R1 |- m26)$) {\color{jxorange}$\mathcal{R}_t^{\mathrm{GFL}}$};
            \node (R3) at ($(R1 |- m36)$) {\color{jxblue}$\mathcal{E}_t^{\mathrm{GFM}}$};
            \node (R4) at ($(R1 |- m46)$) {\color{jxblue}$\mathcal{R}_t^{\mathrm{GFM}}$};
            \node (R5) at ($(R1 |- m56)$) {\color{jxblue}$\mathcal{G}$};
            \node (R6) at ($(R1 |- m66)$) {\color{jxgreen}$\mathcal{L}$};

            \node (m11) at (m-1-1) {};
            \node (m12) at (m-1-2) {};
            \node (m13) at (m-1-3) {};
            \node (m14) at (m-1-4) {};
            \node (m15) at (m-1-5) {};
            \node (m16) at (m-1-6) {};
            \node[above=7pt of m11] (C1) {\color{jxorange}$\mathcal{E}_t^{\mathrm{GFL}}$};
            \node at ($(C1 -| m12)$) {\color{jxorange}$\mathcal{R}_t^{\mathrm{GFL}}$};
            \node at ($(C1 -| m13)$) {\color{jxblue}$\mathcal{E}_t^{\mathrm{GFM}}$};
            \node at ($(C1 -| m14)$) {\color{jxblue}$\mathcal{R}_t^{\mathrm{GFM}}$};
            \node at ($(C1 -| m15)$) {\color{jxblue}$\mathcal{G}$};
            \node at ($(C1 -| m16)$) {\color{jxgreen}$\mathcal{L}$};

            \draw[decorate,decoration={brace,amplitude=5pt}] ($(R1.east) + (3pt, 2pt)$) -- ($(R1.east |- R2.east) + (3pt, -2pt)$);
            \node at ($(R1.east)!0.5!(R2.east)$) [right=8pt] {$\mathcal{I}_t$};
            \draw[decorate,decoration={brace,amplitude=5pt}] ($(R3.east) + (2pt, 2pt)$) -- ($(R3.east |- R6.east) + (2pt, -2pt)$);
            \node at ($(R3.east)!0.5!(R6.east)$) [right=11pt] {$\mathcal{J}_t$};
        \end{tikzpicture}
    \end{equation}
\end{figure*}

\subsection{Notational Conventions}

For a set $\mathcal{T}$, the symbol $\abs{\mathcal{T}}$ denotes its cardinality.
For a vector $v$, we use $\norm{v}_2$ to denote its 2-norm and $v_i$ to denote its $i$-th entry.
The notation $\diag(v)$ denotes a diagonal matrix with $v_i$ as the $i$-th diagonal entry.
For a matrix $A$, the submatrix $A_{\mathcal{S} \mathcal{T}}$ refers to the block of $A$ consisting of rows indexed by the set $\mathcal{S}$ and columns indexed by the set $\mathcal{T}$.
We use $I$ to denote the identity matrix of appropriate size.
Additional notations will be introduced as needed.

Then, we consider a power system with $n_{\mathrm{E}}$ energy storage systems, $n_{\mathrm{R}}$ renewable energy sources, and $n_{\mathrm{G}}$ synchronous generators, where all energy storage systems and renewable energy sources are interfaced through inverters.
Let $\mathcal{E}=\set{1,2,\ldots,n_{\mathrm{E}}}$, $\mathcal{R}=\set{n_{\mathrm{E}}+i\given i=1,2,\ldots,n_{\mathrm{R}}}$, and $\mathcal{G}=\set{n_{\mathrm{E}}+n_{\mathrm{R}}+i\given i=1,2,\ldots,n_{\mathrm{G}}}$ denote the correpsponding index sets.
Then, the bus index set of the entire system is $\mathcal{N}=\mathcal{E}\cup\mathcal{R}\cup\mathcal{G}\cup\mathcal{L}$, where $\mathcal{L}=\set{n_{\mathrm{E}}+n_{\mathrm{R}}+n_{\mathrm{G}}+i\given i=1,2,\ldots,n_{\mathrm{L}}}$ denotes the set of remaining buses without generation resources.
Moreover, we partition $\mathcal{E}=\mathcal{E}_t^{\mathrm{GFM}}\cup\mathcal{E}_t^{\mathrm{GFL}}$ and $\mathcal{R}=\mathcal{R}_t^{\mathrm{GFM}}\cup\mathcal{R}_t^{\mathrm{GFL}}$ according to the operating modes (GFM v.s. GFL) of energy storage systems and renewable energy sources at time $t\in\mathcal{T}$, respectively.
The system admittance matrix, $B^{\mathrm{sys}}$, can be partitioned according to the above bus sets as shown in \eqref{eq:B-ori}, where $\ast$ denotes the symmetric entries.
We further define $\mathcal{I}_t=\mathcal{E}_t^{\mathrm{GFL}}\cup\mathcal{R}_t^{\mathrm{GFL}}$ and $\mathcal{J}_t=\mathcal{N}\setminus\mathcal{I}_t$ for brevity.

\subsection{System Strength-Constrained Scheduling Problem}

We briefly introduce the formulation of the system strength-constrained scheduling problem considered in this work to provide a basic conceptual understanding of its structure.
The complete and detailed formulation is presented in Section~\ref{sec:reformulation-scheduling}, following the development of the core theoretical results.

The objective is to minimize the total system operating cost.
The decision variables include both continuous and binary variables.
The binary variables comprise the GFM/GFL mode-switching variables of IBRs (including both renewable generators and energy storage systems), the commitment-status variables of synchronous generators, and the mutually exclusive charging/discharging status variables of energy storage systems.
The constraints can be broadly classified into three categories.
The first category consists of network constraints, including the DC power flow equations and transmission capacity limits, which are presented in \eqref{eq:network}.
The second category comprises generation-side operating constraints, including renewable generation limits, energy storage charging and discharging constraints, unit commitment constraints for synchronous generators, and GFM/GFL mode-switching constraints for IBRs.
Among these constraints, the unit commitment constraints follow the standard formulations used in conventional unit commitment problems \cite{yonghong2023security}.
The remaining constraints, which are specific to the scheduling framework considered here, are formulated in detail in Section~\ref{sec:formulation}.
The third category is the system strength constraint, which is introduced below and reformulated in Section~\ref{sec:reformulation}.

The generalized operational short-circuit ratio (gOSCR) is adopted as the system strength metric \cite{yuanhui2025placing,chenxi2024generalized}
\begin{equation}
    \label{eq:gOSCR-constraint}
    \gamma \coloneqq \min\lambda^+\left(P^{-1}B\right)\geq \gamma_0,
\end{equation}
where $\lambda^+(\cdot)$ denotes the positive spectrum set of a matrix.
\(P\) is a diagonal matrix composed of the operating power of IBRs in the GFL mode, while \(B\) is the Schur-complement reduction of \(B^{\mathrm{sys}}\) onto the buses associated with GFL-mode IBRs.
The dimensions of \(P\) and \(B\) are both equal to the number of IBRs operating in GFL mode and therefore vary with the GFM/GFL mode-switching decisions.
Time index $t\in\mathcal{T}$ is omitted here for brevity.
Specifically,
\begin{equation}
    \label{eq:P-def}
    P \coloneqq \diag\left([p_{i,t}]_{i\in\mathcal{I}_t}\right),
\end{equation}
\begin{equation}
    \label{eq:B-def}
    B \coloneqq B_{\mathcal{I}_t\mathcal{I}_t}^{\mathrm{sys}}-B_{\mathcal{I}_t\mathcal{J}_t}^{\mathrm{sys}}\left(B_{\mathcal{J}_t\mathcal{J}_t}^{\mathrm{sys}}\right)^{-1}B_{\mathcal{J}_t\mathcal{I}_t}^{\mathrm{sys}},
\end{equation}
where $p_{i,t}$ for $i\in\mathcal{I}_t$ denotes the power injecting to the grid by the corresponding IBR under GFL mode.
We note that $B_{\mathcal{I}_t\mathcal{I}_t}^{\mathrm{sys}}$ and $B_{\mathcal{J}_t\mathcal{J}_t}^{\mathrm{sys}}$ are decision variables implicitly determined by other scheduling decisions such as the on/off status of synchronous generators and the GFM/GFL mode switching of IBRs, while $B_{\mathcal{I}_t\mathcal{J}_t}^{\mathrm{sys}}$ and $B_{\mathcal{J}_t\mathcal{I}_t}^{\mathrm{sys}}$ are constants determined by the network topology and parameters.

\begin{remark}
    The admittance matrix $B^{\mathrm{sys}}$ accounts for the internal admittances of online synchronous generators and grid-forming IBRs, distinguishing it from the one (denoted $B^{\mathrm{pf}}$) used in power flow equations.
    This paper adopts the sign convention that the diagonal entries of both $B^{\mathrm{sys}}$ and $B^{\mathrm{pf}}$ are positive.
    The on/off status of synchronous generators and the switching between GFM and GFL modes of IBRs directly affect the formation of $B$, while the operating power of renewable energy sources (which are always injecting power) and energy storage systems (which can be both charging and discharging) under GFL mode form the diagonal entries of $P$.
    The system strength is strongly coupled with the aforementioned scheduling decisions.
\end{remark}

\subsection{Challenges and the Proposed Solution}

\begin{table*}[!t]
    \centering
    \caption{Challenges and the Corresponding Solutions in System Strength-Constrained Scheduling}\label{tab:challenges}
    \renewcommand{\arraystretch}{1.3}
    \begin{tabular}{m{3.5cm} | >{\centering\arraybackslash}m{4cm} | >{\centering\arraybackslash}m{6cm} | >{\centering\arraybackslash}m{1.5cm}}
        \hline\hline
        \multicolumn{2}{c|}{Identified Challenges}                    & \multicolumn{2}{c}{Proposed Solution}                                                                                                                                                                          \\ \hline\hline     operational constraints                              & formulating operational constraints for IBRs with GFM/GFL mode switching behaviors & introduce mode switching variables $x_{i,t}$ and formulate corresponding constraints such as power output headroom & Section~\ref{sec:formulation-switching}  \\ \hline
        \multirow{5.5}{*}{\parbox{3.5cm}{system strength constraint}} & $P,B$ are nonlinear with respect to other scheduling decision variables  & an equivalent linearization technique without introducing additional binary variables & Section~\ref{sec:formulation-linearization} \\ \cline{2-4}
                                                                      & non-convex matrix product $P^{-1}B$                                      & a separation method based on the proposed Theorem~\ref{thm:1}                         & \multirow{3.7}{*}{\shortstack[c]{
        Section~\ref{sec:reformulation-exact}                                                                                                                                                                                                                                          \\[3pt]
                Section~\ref{sec:reformulation-scheduling}
        }}                                                                                                                                                                                                                                                                             \\ \cline{2-3}
                                                                      & varying dimensions of $P,B$ depend on the values of scheduling decisions & a fixed-dimension reformulation based on the proposed Theorem~\ref{thm:2}             &                                             \\ \hline
        entire scheduling problem                                     & should be solver-friendly for mainstream commercial solvers              & a Rayleigh Cut solution method compatible with standard MILP solvers                  & Section~\ref{sec:reformulation-rayleigh}    \\ \hline\hline
    \end{tabular}
\end{table*}

Incorporating the system strength constraint \eqref{eq:gOSCR-constraint} into the scheduling problem, while crucial, presents significant challenges due to its non-convexity and implicit dependence on various scheduling decisions.
We identify three main challenges as follows.

First, renewable energy sources and energy storage systems operating under GFL mode and GFM mode have distinct characteristics.
There is additional headroom limits for GFM mode, which complicates the formulation of operational constraints for IBRs with GFM/GFL mode switching behaviors.
Second, the system strength constraint \eqref{eq:gOSCR-constraint} involves the minimum positive eigenvalue of a matrix product $P^{-1}B$.
The decision variables $P$ and $B$ are nonlinear to other scheduling decision variables, and the minimum positive eigenvalue function is non-convex and has no explicit expression.
Moreover, the dimensions of $P$ and $B$ are determined by the index sets $\mathcal{I}_t$ and $\mathcal{J}_t$, which vary with different GFM/GFL decisions.
It is nearly impossible to have the dimension of a matrix-type decision variable automatically adjust according to the values of other decision variables in mathematical optimizations.
Third, the whole system strength-constrained scheduling problem is required to be solver-friendly such that it can be handled by mainstream commercial solvers, which is non-trivial given the aforementioned challenges.

\tabref{tab:challenges} summarizes the identified challenges and lists the corresponding proposed solutions in this paper.
Regarding the first challenge, we introduce mode switching variables $x_{i,t}$ for each IBR and formulate corresponding constraints such as power output headroom to effectively capture the operational characteristics of IBRs under GFM/GFL mode switching behaviors (see Section~\ref{sec:formulation-switching}).
Regarding the second challenge, we derive an equivalent linearization technique without introducing additional binary variables to handle the nonlinearity of $P$ and $B$ with respect to other scheduling decision variables (see Section~\ref{sec:formulation-linearization}).
Then, we propose a separation method based on Theorem~\ref{thm:1} to handle the non-convex matrix product $P^{-1}B$, and a fixed-dimension reformulation based on Theorem~\ref{thm:2} to handle the varying dimensions of $P$ and $B$ (see Section~\ref{sec:reformulation-exact} and Section~\ref{sec:reformulation-scheduling}).
Finally, we provide a Rayleigh Cut solution method compatible with standard MILP solvers to ensure the solver-friendliness of the entire scheduling problem (see Section~\ref{sec:reformulation-rayleigh}).

\section{Formulation of the Operating Constraints for IBRs with GFM/GFL Mode Switching\label{sec:formulation}}

In this section, we formulate the operating behavior of IBRs under GFM and GFL modes.
The device-level operational constraints are concretized by describing the switching between GFM and GFL modes of IBRs and the corresponding power output constraints.
An equivalent linearization of these operational constraints is then provided to address the nonlinearity brought by mode switching.

\subsection{Switching Between GFM and GFL Modes\label{sec:formulation-switching}}

We define a binary variable $x_{i,t}$ for all $i\in\mathcal{E}\cup\mathcal{R}$ to indicate the GFM/GFL mode of IBRs and for all $i\in\mathcal{G}$ to indicate the on/off status of synchronous generators as follows:
\begin{equation}
    \begin{split}
        x_{i,t} & =0 \iff i\in\mathcal{E}^{\mathrm{GFL}}_t\cup\mathcal{R}^{\mathrm{GFL}}_t\cup\mathcal{G}^{\mathrm{OFF}}_t \\
        x_{i,t} & =1 \iff i\in\mathcal{E}^{\mathrm{GFM}}_t\cup\mathcal{R}^{\mathrm{GFM}}_t\cup\mathcal{G}^{\mathrm{ON}}_t
    \end{split}
\end{equation}
where $\mathcal{G}^{\mathrm{OFF}}_t$ and $\mathcal{G}^{\mathrm{ON}}_t$ denote the index sets of offline and online synchronous generators at time $t$, respectively.

IBRs operating in GFM mode need to reserve certain capacity to remain grid-forming capability.
Thus, their maximum power limits vary with the GFM/GFL mode.
We formulate a unified constraint to describe the maximum power limit $\widetilde{p_{i,t}}$ of a IBR under GFM/GFL modes as follows:
\begin{equation}
    \label{eq:maximum-power}
    \widetilde{p_{i,t}} = p_{i,t}^{\max} - \alpha_i x_{i,t},\quad \forall i\in\mathcal{E}\cup\mathcal{R},
\end{equation}
where $\alpha_i$ denotes the reserved capacity for GFM operation.
For energy storage $i\in\mathcal{E}$, $p_{i,t}^{\max}$ indicates its rated power, which is constant over all time horizons, i.e., $p_{i,t}=p_{i,1}\eqqcolon p_i^{\max}$ for all $t\in\mathcal{T}$.
For renewable energy source $i\in\mathcal{R}$, $p_{i,t}^{\max}$ indicates its total available power, which usually varies with time $t\in\mathcal{T}$.

For renewable energy sources, their scheduled power $p_{i,t}$ are constrained by the maximum power limit.
\begin{equation}
    \label{eq:renewable-power}
    0\leq p_{i,t}\leq \widetilde{p_{i,t}},\quad\forall i\in\mathcal{R}.
\end{equation}
\begin{remark}
    For renewable energy source $i\in\mathcal{R}$, if its total available power $p_{i,t}$ is less than the reserved capacity for GFM operation $\alpha_i$, i.e., $p_{i,t}^{\max}<\alpha_i$, it cannot operate in GFM mode.
    The constraints \eqref{eq:maximum-power} and \eqref{eq:renewable-power} can implicitly ensure that $x_{i,t}=0$ in this case.
\end{remark}

For energy storage systems, their scheduled power $p_{i,t}$ can be either positive (discharging) or negative (charging).
\begin{equation}
    \label{eq:energy-storage-output}
    p_{i,t} = p^{\mathrm{dis}}_{i,t} - p^{\mathrm{cha}}_{i,t},\quad\forall i\in\mathcal{E}.
\end{equation}
The charging and discharging powers are constrained by the maximum power limit.
\begin{subequations}
    \label{eq:cha-dis}
    \begin{align}
         & 0\leq                       p^{\mathrm{dis}}_{i,t}\leq \widetilde{p_{i,t}}\delta_{i,t} \\
         & 0\leq p^{\mathrm{cha}}_{i,t} \leq \widetilde{p_{i,t}}-\widetilde{p_{i,t}}\delta_{i,t}
    \end{align}
\end{subequations}
where $\delta_{i,t}\in\set{0,1}$ is a binary variable indicating whether the energy storage is discharging ($\delta_{i,t}=1$) or charging ($\delta_{i,t}=0$) at time $t$.
Besides, the energy storage systems are constrained by their energy capacity and energy balance:
\begin{subequations}
    \label{eq:energy-storage-soc}
    \begin{align}
        \beta_i x_{i,t}  \leq                              & E_{i,t}\leq E_i^{\max}  - \beta_i x_{i, t}                                                                                           \\
                                                           & E_{i,0} = E_{i,\abs{\mathcal{T}}}                                                                                                    \\
        E_{i,t} =  \left(1-\eta^{\mathrm{self}}_{i}\right) & E_{i,t-1} + \left(\eta_i^{\mathrm{cha}} p^{\mathrm{cha}}_{i,t} - \frac{p^{\mathrm{dis}}_{i,t}}{\eta_i^{\mathrm{dis}}}\right)\Delta t
    \end{align}
\end{subequations}
where $E_{i,t}$ denotes stored energy, $E_i^{\max}$ denotes the maximum energy capacity, $\beta_i$ denotes the energy reserve for GFM operation, $\eta^{\mathrm{self}}_{i}$ denotes the self-discharge rate, $\eta_i^{\mathrm{cha}}$ and $\eta_i^{\mathrm{dis}}$ denote the charging and discharging efficiencies, and $\Delta t$ denotes the time interval between two consecutive time steps.

For synchronous generators, their scheduled power $p_{i,t}$ are constrained by classic unit commitment constraints, which are omitted here for brevity.

Finally, we note that the decision variables $P$ and $B$ in \eqref{eq:gOSCR-constraint} are related to the switching decision variable $x_{i,t}$.
In detail, $P$ only includes the scheduled power of IBRs operating in GFL mode:
\begin{equation}
    \label{eq:P-x}
    \hat{P}=\diag\left(\begin{bmatrix}(1-x_{i,t})p_{i,t}\end{bmatrix}_{i\in\mathcal{E}\cup\mathcal{R}}\right).
\end{equation}
The decision variable $B$ is equal to the Kron-reduced admittance matrix $B^{\mathrm{sys}}$, according to \eqref{eq:B-ori}.
The elements of $B^{\mathrm{sys}}$ has the following relationship with $x_{i,t}$:
\begin{equation}
    \label{eq:Bsys-x}
    B^{\mathrm{sys}}_{ij} = \begin{cases}
        B^{\mathrm{pf}}_{ii} - b_i x_{i,t}\quad     & i=j\in\mathcal{E}\cup\mathcal{R}\cup\mathcal{G} \\
        B^{\mathrm{pf}}_{ii}(=\mathrm{const.})\quad & i=j\in\mathcal{L}                               \\
        B^{\mathrm{pf}}_{ij}(=\mathrm{const.})\quad & i\neq j
    \end{cases}
\end{equation}
where $B^{\mathrm{pf}}$ is the admittance matrix correpsponding to the power flow calculation, and $b_i$ denotes the internal admittance of power source $i$.
We have $b_i<0$ for inductive sources.

\begin{remark}
    \label{rmk:size-change-issues}
    The variable $\hat{P}$ in constraint \eqref{eq:P-x} is not exactly the same as $P$ since the size of $\hat{P}$ has to shrink when some IBRs switch to GFM mode, otherwise the inverse operation in \eqref{eq:gOSCR-constraint} is not well-defined.
    It is similar for the decision variable $B$ defined in \eqref{eq:B-def} since the partition of $B^{\mathrm{sys}}$ also depends on the GFL/GFM mode switching.
    However, we temporarily leave the size-change issue aside for clarity and will address them in Section \ref{sec:reformulation}.
    The subscript $t$ is omitted here for brevity.
\end{remark}

\subsection{Equivalent Linearization of the Operational Constraints\label{sec:formulation-linearization}}

The operational constraints formulated in the previous subsection involve bilinear terms and matrix inverse operations due to the mode switching decision variables (and on/off decision variables) $x_{i,t}$.
In this subsection, we introduce equivalent linearization techniques to handle these bilinear terms.

First, we linearize the bilinear term $\widetilde{p_{i,t}}\delta_{i,t}$ in energy storage systems' operational constraint \eqref{eq:cha-dis}.
According to \eqref{eq:maximum-power}, it is equivalent to linearize the bilinear term $x_{i,t}\delta_{i,t}$, which can be achieved by introducing an auxiliary continuous variable $w_{i,t}$.
The linearization is shown in \eqref{eq:x-delta}.
\begin{equation}
    \label{eq:x-delta}
    \begin{cases}
        w_{i,t}=x_{i,t}\delta_{i,t} \\
        x_{i,t},\delta_{i,t}\in\set{0,1}
    \end{cases}\hspace{-8pt}
    \iff \begin{cases}
        w_{i,t}\leq x_{i,t}                    \\
        0\leq w_{i,t}\leq \delta_{i,t}         \\
        x_{i,t} + \delta_{i,t} - w_{i,t}\leq 1 \\
        x_{i,t},\delta_{i,t}\in\set{0,1}
    \end{cases}
\end{equation}

Then, we introduce an auxiliary continuous variable $\rho_{i,t}$ linearize the bilinear terms $x_{i,t}p_{i,t}$ in \eqref{eq:P-x}.
\begin{equation}
    \label{eq:linaerize-P}
    \begin{cases}
        \rho_{i,t}=x_{i,t}p_{i,t} \\
        x_{i,t}\in\set{0,1}
    \end{cases}\hspace{-8pt}
    \!\iff \!\begin{cases}
        \rho_{i,t}\leq p_{i}^{\max}                          \\
        \rho_{i,t}\geq -p_i^{\max}                           \\
        \rho_{i,t}\leq p_{i,t} + p_i^{\max}(1\! -\! x_{i,t}) \\
        \rho_{i,t}\geq p_{i,t} - p_i^{\max}(1\! -\! x_{i,t}) \\
        x_{i,t}\in\set{0,1}
    \end{cases}
\end{equation}

Regarding the matrix inverse operation in \eqref{eq:B-def}, we introduce an auxiliary matrix variable $A_t=[A_{ij,t}]_{i,j\in\mathcal{E}\cup\mathcal{R}\cup\mathcal{G}}$ to represent this inverse result at time $t$.
Combining the definition of matrix inverse and \eqref{eq:Bsys-x}, we formulate the following equivalent constraints that $A$ must satisfy, where the indices $i,j,k$ all tranverse $\mathcal{E}\cup\mathcal{R}\cup\mathcal{G}$.
\begin{equation}
    \begin{cases}
        -b_{i}x_{i,t}A_{ii,t}+\sum_{k} B_{ik}^{\mathrm{pf}}A_{kj,t} = 1, & i=j     \\
        -b_{i}x_{i,t}A_{ij,t}+\sum_{k} B_{ik}^{\mathrm{pf}}A_{kj,t} = 0, & i\neq j \\
    \end{cases}
\end{equation}
The bilinear terms $x_{i,t}A_{ij,t}$ can be linearized by introducing auxiliary continuous variables $z_{ij,t}$, as shown in \eqref{eq:x-A}.
\begin{equation}
    \label{eq:x-A}
    \begin{cases}
        z_{ij,t}=x_{i,t}A_{ij,t} \\
        x_{i,t}\in\set{0,1}
    \end{cases}\hspace{-15pt}
    \iff
    \hspace{-5pt}
    \begin{cases}
        z_{ij,t} \leq Mx_{i,t}                  \\
        z_{ij,t} \geq -Mx_{i,t}                 \\
        z_{ij,t} \leq A_{ij,t} + M(1 - x_{i,t}) \\
        z_{ij,t} \geq A_{ij,t} - M(1 - x_{i,t}) \\
        x_{i,t}\in\set{0,1}
    \end{cases}
\end{equation}
where $M\geq\max_{i,j,t}\abs{A_{ij,t}}$ is a large positive constant.

\begin{remark}
    (1) The linearizations are exact and do not bring any approximation error.
    (2) The linearization techniques introduced in this subsection do not introduce additional binary decision variables except for the original ones.
    (3) The value of $M$ can be easily estimated and set in practical applications.
\end{remark}

\section{Integrating the System Strength Constraint in the Scheduling Problem\label{sec:reformulation}}

In this section, we integrate the system strength constraint into the scheduling model.
First, we theoretically show that the system strength constraint \eqref{eq:gOSCR-constraint} can be equivalently reformulated as a LMI \eqref{eq:final-lmi} where the size change issues and non-convex issues are both addressed.
Then, we convert the original non-convex system strength-constrained scheduling problem into a MISDP problem by incorporating the reformulated system strength constraint along with the operational constraints derived in Section \ref{sec:formulation}.
Finally, we further provide a Rayleigh Cut method to solve the converted system strength-constrained scheduling problem.

\subsection{Exact Reformulation of the System Strength Constraint\label{sec:reformulation-exact}}

As explained in Remark~\ref{rmk:size-change-issues}, the decision variables $P$ and $B$ in \eqref{eq:gOSCR-constraint} have size-change issues due to the mode switching decision variables $x_{i,t}$.
In detail, the size of $P$ and $B$ depends on the index set $\mathcal{I}_t$ that is a function of decision variables $x_{i,t}$:
\begin{equation}
    \mathcal{I}_t = \set{i\in\mathcal{E}\cup\mathcal{R}\given x_{i,t}=0}.
\end{equation}
Such size-change issues make the system strength constraint difficult to handle in optimization problems because the size of decision variables is unknown before solving the optimization problems and even varies for different feasible solutions.
In addition, the system strength constraint \eqref{eq:gOSCR-constraint} is non-convex and has no explicit expression due to minimum positive eigenvalue operation.

To address these challenges, we show the following theorems to reformulate the system strength constraint \eqref{eq:gOSCR-constraint} into a LMI \eqref{eq:final-lmi} with fixed size.
The reformulation is exact and does not introduce any approximation error, providing a solid theoretical foundation for integrating the system strength constraint into scheduling problems.

\begin{theorem}
    \label{thm:1}
    Let $P$ and $B$ be the same as those defined in \eqref{eq:P-def} and \eqref{eq:B-def}, respectively.
    The system strength constraint \eqref{eq:gOSCR-constraint} is equivalent to the following LMI:
    \begin{equation}
        \label{eq:system_strength_lmi}
        B - \gamma_0 P \succeq 0.
    \end{equation}
\end{theorem}
\begin{proof}
    ``$\implies$'':
    Since $B\succ 0$, there exists $B^{\frac{1}{2}}\succ 0$ such that $B=B^{\frac{1}{2}}B^{\frac{1}{2}}$.
    By multiplying $B^{-\frac{1}{2}}$ and $B^{\frac{1}{2}}$ on the left and right sides of $P^{-1}B$, respectively, we have
    \begin{equation}
        \lambda(P^{-1}B) = \lambda(B^{\frac{1}{2}}P^{-1}B^{\frac{1}{2}}).
    \end{equation}
    Since the eigenvalues of a non-singular matrix is the reciprocal of those of its inverse, we see that \eqref{eq:gOSCR-constraint} is equivalent to
    \begin{equation}
        I-\gamma_0B^{-\frac{1}{2}}PB^{-\frac{1}{2}}\succeq 0.
    \end{equation}
    By multiplying $B^{\frac{1}{2}}\succ 0$ on both sides, we obtain \eqref{eq:system_strength_lmi}.

    ``$\impliedby$'':
    For all $\xi\neq 0$, the condition $B-\gamma_0 P\succeq 0$ implies $        \xi^{\T} B \xi \geq \gamma_0 \xi^{\T} P \xi$, which further gives
    \begin{equation}
        \gamma_0\leq\min_{\xi\in\Xi}\frac{\xi^{\T} B \xi}{\xi^{\T} P \xi},\quad \Xi\coloneqq\set{\xi\given \xi^{\T} P \xi>0, \norm{\xi}_2=1}.
    \end{equation}
    On the other hand, let $\gamma:=\min\lambda^+\left(P^{-1}B\right)$, there exists $\zeta\neq 0$ with $\norm{\zeta}_2=1$ such that $P^{-1}B\zeta=\gamma\zeta$, which implies
    \begin{equation}
        \zeta^{\T} P\zeta = \frac{1}{\gamma}\zeta^{\T} B\zeta>0\implies \zeta\in\Xi,
    \end{equation}
    where the last inequality holds since $\gamma>0$ and $B\succ 0$.
    Therefore,
    \begin{equation}
        \gamma_0\leq\min_{\xi\in\Xi}\frac{\xi^{\T} B \xi}{\xi^{\T} P \xi}\leq \frac{\zeta^{\T} B \zeta}{\zeta^{\T} P \zeta}=\gamma,
    \end{equation}
    which implies \eqref{eq:gOSCR-constraint}.
\end{proof}

\begin{theorem}
    \label{thm:2}
    Let $P,B,B^\mathrm{sys}$ be the same as those defined in \eqref{eq:P-def}-\eqref{eq:B-def} and \eqref{eq:Bsys-x}, respectively.
    For a given $\gamma_0>0$, the LMI $B-\gamma_0 P \succeq 0$ is equivalent to the following LMI:
    \begin{equation}
        \label{eq:final-lmi-0}
        \begin{bmatrix}
            B^{\mathrm{sys}}_{\mathcal{I}_t\mathcal{I}_t} & B^{\mathrm{sys}}_{\mathcal{I}_t\mathcal{J}_t} \\[2pt]
            B^{\mathrm{sys}}_{\mathcal{J}_t\mathcal{I}_t} & B^{\mathrm{sys}}_{\mathcal{J}_t\mathcal{J}_t}
        \end{bmatrix} - \gamma_0 \begin{bmatrix}
            P &   \\
              & 0
        \end{bmatrix}\succeq 0.
    \end{equation}
\end{theorem}
\begin{proof}
    ``$\implies$'':
    The condition $B-\gamma_0 P\succeq 0$ implies that
    \begin{equation}
        \label{eq:thm2-1}
        \xi^{\T}\left(B^{\mathrm{sys}}_{\mathcal{I}_t\mathcal{I}_t} - \gamma_0 P\right)\xi \geq \xi^{\T}B^{\mathrm{sys}}_{\mathcal{I}_t\mathcal{J}_t}\left(B^{\mathrm{sys}}_{\mathcal{J}_t\mathcal{J}_t}\right)^{-1}B^{\mathrm{sys}}_{\mathcal{J}_t\mathcal{I}_t}\xi,
    \end{equation}
    for all $\xi$.
    Then, for all $\xi,\zeta$, we have
    \begin{equation}
        \begin{split}
             & \begin{bmatrix}
                   \xi \\
                   \zeta
               \end{bmatrix}^{\T}
            \begin{bmatrix}
                B^{\mathrm{sys}}_{\mathcal{I}_t\mathcal{I}_t} - \gamma_0 P & B^{\mathrm{sys}}_{\mathcal{I}_t\mathcal{J}_t} \\
                B^{\mathrm{sys}}_{\mathcal{J}_t\mathcal{I}_t}              & B^{\mathrm{sys}}_{\mathcal{J}_t\mathcal{J}_t}
            \end{bmatrix}\begin{bmatrix}
                             \xi \\
                             \zeta
                         \end{bmatrix} \\
             & =        \begin{bmatrix}
                            \xi \\
                            \zeta
                        \end{bmatrix}^{\T}
            \begin{bmatrix}
                0                                             & B^{\mathrm{sys}}_{\mathcal{I}_t\mathcal{J}_t} \\
                B^{\mathrm{sys}}_{\mathcal{J}_t\mathcal{I}_t} & B^{\mathrm{sys}}_{\mathcal{J}_t\mathcal{J}_t}
            \end{bmatrix}\begin{bmatrix}
                             \xi \\
                             \zeta
                         \end{bmatrix}    +       \xi^{\T}B^{\mathrm{sys}}_{\mathcal{I}_t\mathcal{I}_t}\xi  - \gamma_0 \xi^{\T}P\xi              \\
             & \geq \begin{bmatrix}
                        \xi \\
                        \zeta
                    \end{bmatrix}^{\T}
            \begin{bmatrix}
                B^{\mathrm{sys}}_{\mathcal{I}_t\mathcal{J}_t}\left(B^{\mathrm{sys}}_{\mathcal{J}_t\mathcal{J}_t}\right)^{-1}B^{\mathrm{sys}}_{\mathcal{J}_t\mathcal{I}_t} & B^{\mathrm{sys}}_{\mathcal{I}_t\mathcal{J}_t} \\
                B^{\mathrm{sys}}_{\mathcal{J}_t\mathcal{I}_t}                                                                                                             & B^{\mathrm{sys}}_{\mathcal{J}_t\mathcal{J}_t}
            \end{bmatrix}\begin{bmatrix}
                             \xi \\
                             \zeta
                         \end{bmatrix}\geq 0,
        \end{split}
    \end{equation}
    in which the first inequality holds by substituting \eqref{eq:thm2-1} and the second inequality holds due to the Schur complement condition.

    ``$\impliedby$'':
    For all $\xi$, let $\zeta =-\left(B^{\mathrm{sys}}_{\mathcal{J}_t\mathcal{J}_t}\right)^{-1}B^{\mathrm{sys}}_{\mathcal{J}_t\mathcal{I}_t}\xi$.
    We have
    \begin{equation}
        \begin{split}
             & 0 \leq  \begin{bmatrix}
                           \xi \\
                           \zeta
                       \end{bmatrix}^{\T}
            \left(\begin{bmatrix}
                          B^{\mathrm{sys}}_{\mathcal{I}_t\mathcal{I}_t} & B^{\mathrm{sys}}_{\mathcal{I}_t\mathcal{J}_t} \\[2pt]
                          B^{\mathrm{sys}}_{\mathcal{J}_t\mathcal{I}_t} & B^{\mathrm{sys}}_{\mathcal{J}_t\mathcal{J}_t}
                      \end{bmatrix} - \gamma_0 \begin{bmatrix}
                                                   P &   \\
                                                     & 0
                                               \end{bmatrix}\right)\begin{bmatrix}
                                                                   \xi \\
                                                                   \zeta
                                                               \end{bmatrix}                                                                                      \\
             & =\xi^{\T}\left(B^{\mathrm{sys}}_{\mathcal{I}_t\mathcal{I}_t} - \gamma_0 P\right)\xi - \xi^{\T}B^{\mathrm{sys}}_{\mathcal{I}_t\mathcal{J}_t}\left(B^{\mathrm{sys}}_{\mathcal{J}_t\mathcal{J}_t}\right)^{-1}B^{\mathrm{sys}}_{\mathcal{J}_t\mathcal{I}_t}\xi,
        \end{split}
    \end{equation}
    which implies that $B-\gamma_0 P\succeq 0$.
\end{proof}

\begin{corollary}
    \label{cor:1}
    Denote $\mathcal{I}=\mathcal{E}\cup\mathcal{R}$ and $\mathcal{J}=\mathcal{G}\cup\mathcal{L}$ and let $\hat{P}$ and $B^{\mathrm{sys}}$ be the same as those defined in \eqref{eq:P-x} and \eqref{eq:Bsys-x}, respectively.
    The system strength constraint \eqref{eq:gOSCR-constraint} is equivalent to the following LMI:
    \begin{equation}
        \label{eq:final-lmi}
        \hat{B} - \gamma_0 \hat{P} \succeq 0,
    \end{equation}
    where $\hat{B}:=B^{\mathrm{sys}}_{\mathcal{II}} -B^{\mathrm{sys}}_{\mathcal{IJ}}\left(B^{\mathrm{sys}}_{\mathcal{JJ}}\right)^{-1}B^{\mathrm{sys}}_{\mathcal{JI}}$.
\end{corollary}
\begin{proof}
    Combining Theorems~\ref{thm:1}-\ref{thm:2} and the Kron reduction, we obtain the conclusion.
\end{proof}

\begin{remark}
    The LMI \eqref{eq:final-lmi} has a fixed size $n_\mathrm{E}+n_\mathrm{R}$ regardless of the mode switching decision variables $x_{i,t}$.
    In addition, its expression is explicit since both $\hat{P}$ and $\hat{B}$ can be directly constructed by \eqref{eq:P-x}-\eqref{eq:Bsys-x} and \eqref{eq:x-A} without any inverse operation.
    Thus, the size-change issues and non-convex issues in the original system strength constraint \eqref{eq:gOSCR-constraint} are both addressed in \eqref{eq:final-lmi}, facilitating its integration into scheduling problems.
\end{remark}

Therefore, the original system strength constraint \eqref{eq:gOSCR-constraint} is equivalently reformulated as the LMI \eqref{eq:final-lmi}, which has a fixed size and explicit expression.

\begin{remark}
    \jxrevise{To the best of our knowledge, the system-strength-constrained scheduling model proposed in this paper is the first tractable formulation that captures the complex coupling between system strength constraints and operational decisions, such as GFM/GFL mode switching, without introducing additional approximation errors.
        This tractability is enabled by the theoretical reformulation of the system strength constraint established in Theorems~\ref{thm:1}-\ref{thm:2}.
        Therefore, the proposed reformulation constitutes a major theoretical contribution of this work, rather than a minor incremental extension of existing studies.}
\end{remark}

\subsection{The Complete Scheduling Model\label{sec:reformulation-scheduling}}

We now specify the complete system strength-constrained scheduling model by incorporating the reformulated system strength constraint \eqref{eq:final-lmi} into the scheduling problem.
The objective is to minimize the total operating cost, which consists of three parts: the operating cost of synchronous generators $C_1$, the operating cost of renewable energy sources $C_2$, and the operating cost of energy storage systems $C_3$.
\begin{subequations}
    \begin{align}
        C_1 & = \sum_{i\in\mathcal{G},t\in\mathcal{T}} c^{\mathrm{gen}}_{i,t} p_{i,t} + c^{\mathrm{fix}}_i x_{i,t} + c^{\mathrm{up}}_i su_{i,t} + c^{\mathrm{dn}}_i sd_{i,t}          \\
        C_2 & = \sum_{i\in\mathcal{R},t\in\mathcal{T}} c^{\mathrm{gen}}_{i,t} p_{i,t} + c^{\mathrm{cur}}_{i,t} (\widetilde{p_{i,t}} - p_{i,t}) + c^{\mathrm{gfm}}_{i,t} x_{i,t}       \\
        C_3 & = \sum_{i\in\mathcal{E},t\in\mathcal{T}} c^{\mathrm{dis}}_{i,t} p_{i,t}^{\mathrm{dis}} + c^{\mathrm{cha}}_{i,t} p_{i,t}^{\mathrm{cha}} + c^{\mathrm{gfm}}_{i,t} x_{i,t}
    \end{align}
\end{subequations}
Where $su_{i,t}$ and $sd_{i,t}$ denote the startup and shutdown indicators of synchronous generator $i$ at time $t$, and all parameters denoted by $c$ represent the corresponding cost coefficients.
The four terms in $C_1$ denote the variable generation cost, online fixed cost, startup cost, and shutdown cost of synchronous generators, respectively.
The three terms in $C_2$ denote the variable generation cost, curtailment cost, and additional cost for GFM operation of renewable energy sources, respectively.
The three terms in $C_3$ denote the discharging cost, charging cost, and additional cost for GFM operation of energy storage systems, respectively.

The DC power flow equations are used to model the network constraints, as shown in \eqref{eq:dc-power-balance} and \eqref{eq:dc-line-flow}.
\begin{subequations}
    \label{eq:network}
    \begin{align}
        \sum_{j\in\mathcal{N}}B^{\mathrm{pf}}_{ij}\theta_{j,t}  = p_{i,t} - d_{i,t},\quad & \forall i\in\mathcal{N},t\in\mathcal{T} \label{eq:dc-power-balance}                  \\
        \abs*{B^{\mathrm{pf}}_{ij}(\theta_{i,t}-\theta_{j,t})} \leq f_{ij}^{\max},\quad   & \forall (i,j)\in\mathcal{N}\times\mathcal{N},t\in\mathcal{T} \label{eq:dc-line-flow}
    \end{align}
\end{subequations}
Where $\theta_{i,t}$ denotes the voltage angle at bus $i$ and time $t$, $d_{i,t}$ denotes the load demand at bus $i$ and time $t$, and $f_{ij}^{\max}$ denotes the maximum power flow limit of line $(i,j)$.

The operating constraints of synchronous generators are given in classic unit commitment constraints, including minimum technical generation limits, minimum up/down time limits, etc., which are omitted in this paper for brevity.
We use $\text{\emph{UC}}$ to represent these constraints.
The operating constraints of renewable energy sources are given in \eqref{eq:maximum-power} and \eqref{eq:renewable-power}.
The operating constraints of energy storage systems are given in \eqref{eq:maximum-power}, \eqref{eq:energy-storage-output}, \eqref{eq:cha-dis}, \eqref{eq:energy-storage-soc}, and \eqref{eq:x-delta}.
The system strength constraints are given in \eqref{eq:P-x}, \eqref{eq:Bsys-x}, \eqref{eq:x-A}, and \eqref{eq:final-lmi}.

Therefore, the complete scheduling model is formulated as follows:
\begin{subequations}
    \label{eq:complete-scheduling}
    \begin{align}
        \min\quad        & C_1 + C_2 + C_3                                                                                                                                                                 \\
        \text{s.t.}\quad & \eqref{eq:maximum-power}, \eqref{eq:renewable-power}, \eqref{eq:energy-storage-output}, \eqref{eq:cha-dis}, \eqref{eq:energy-storage-soc}, \eqref{eq:x-delta}, \text{\emph{UC}} \\
                         & \eqref{eq:network}                                                                                                                                                              \\
                         & \eqref{eq:P-x}, \eqref{eq:Bsys-x}, \eqref{eq:x-A}, \eqref{eq:final-lmi}
    \end{align}
\end{subequations}

\begin{remark}
    Through rigorous theoretical derivations and proofs, we have successfully converted the original non-convex and non-explicit scheduling problem into a MISDP problem \eqref{eq:complete-scheduling} without any approximation errors.
\end{remark}

\subsection{Rayleigh Cut Solution Method}\label{sec:reformulation-rayleigh}

The MISDP problems is still intractable and not solver-friendly in general.
Therefore, we further provide a Rayleigh Cut method to solve the MISDP problem \eqref{eq:complete-scheduling} in an iterative manner, which guarantees feasibility and optimality.

First, we show that the system strength constraint \eqref{eq:final-lmi} is equivalent to a set of linear inequalities \eqref{eq:Rayleigh-cut}.
Each constraint in \eqref{eq:Rayleigh-cut} can be viewed as a cutting plane that removes infeasible solutions violating \eqref{eq:final-lmi}, thereby named Rayleigh Cut.

\begin{theorem}
    \label{thm:Rayleigh}
    The LMI \eqref{eq:final-lmi} is equivalent to the following set of linear inequalities:
    \begin{equation}
        \label{eq:Rayleigh-cut}
        u^{\T}\left(\hat{B}-\gamma_0\hat{P}\right)u\geq 0,\ \forall u\in\bar{\mathcal{U}},
    \end{equation}
    where $\bar{\mathcal{U}}$ is defined as
    \begin{equation}
        \bar{\mathcal{U}}:=\set*{u\given u=\psi(\hat{B}-\gamma_0\hat{P})},
    \end{equation}
    in which $\psi(\cdot)$ denotes the unit-length eigenvector corresponding to the minimum eigenvalue of its matrix argument.
\end{theorem}
\begin{proof}
    ``$\implies$'':
    This direction is straightforward by the definition of positive semidefinite matrices.

    ``$\impliedby$'':
    Let $(\mu, u)$ be the minimum eigenpair of $\hat{B}-\gamma_0\hat{P}$ with $\norm{u}_2=1$.
    \begin{equation}
        0\leq u^{\T}\left(\hat{B}-\gamma_0\hat{P}\right)u = \mu \left(u^{\T}u\right) = \mu,
    \end{equation}
    which further implies that $\hat{B}-\gamma_0\hat{P}\succeq 0$.
\end{proof}

Then, we relax the MISDP problem \eqref{eq:complete-scheduling} into a MILP problem by replacing the system strength constraint \eqref{eq:final-lmi} with a part of Rayleigh Cuts, as shown in \eqref{eq:relax-scheduling}, where $\mathcal{U}\subset\bar{\mathcal{U}}$.
\begin{subequations}
    \label{eq:relax-scheduling}
    \begin{align}
        \min\         & C_1 + C_2 + C_3                                                                                                                                                                 \\
        \text{s.t.}\  & \eqref{eq:maximum-power}, \eqref{eq:renewable-power}, \eqref{eq:energy-storage-output}, \eqref{eq:cha-dis}, \eqref{eq:energy-storage-soc}, \eqref{eq:x-delta}, \text{\emph{UC}} \\
                      & \eqref{eq:network}                                                                                                                                                              \\
                      & \eqref{eq:P-x}, \eqref{eq:Bsys-x}, \eqref{eq:x-A}, u^{\T}\left(\hat{B}-\gamma_0\hat{P}\right)u\geq 0,\ \forall u\in \mathcal{U}
    \end{align}
\end{subequations}

\begin{figure}[!t]
    \centering
    \includegraphics[width=0.99\linewidth]{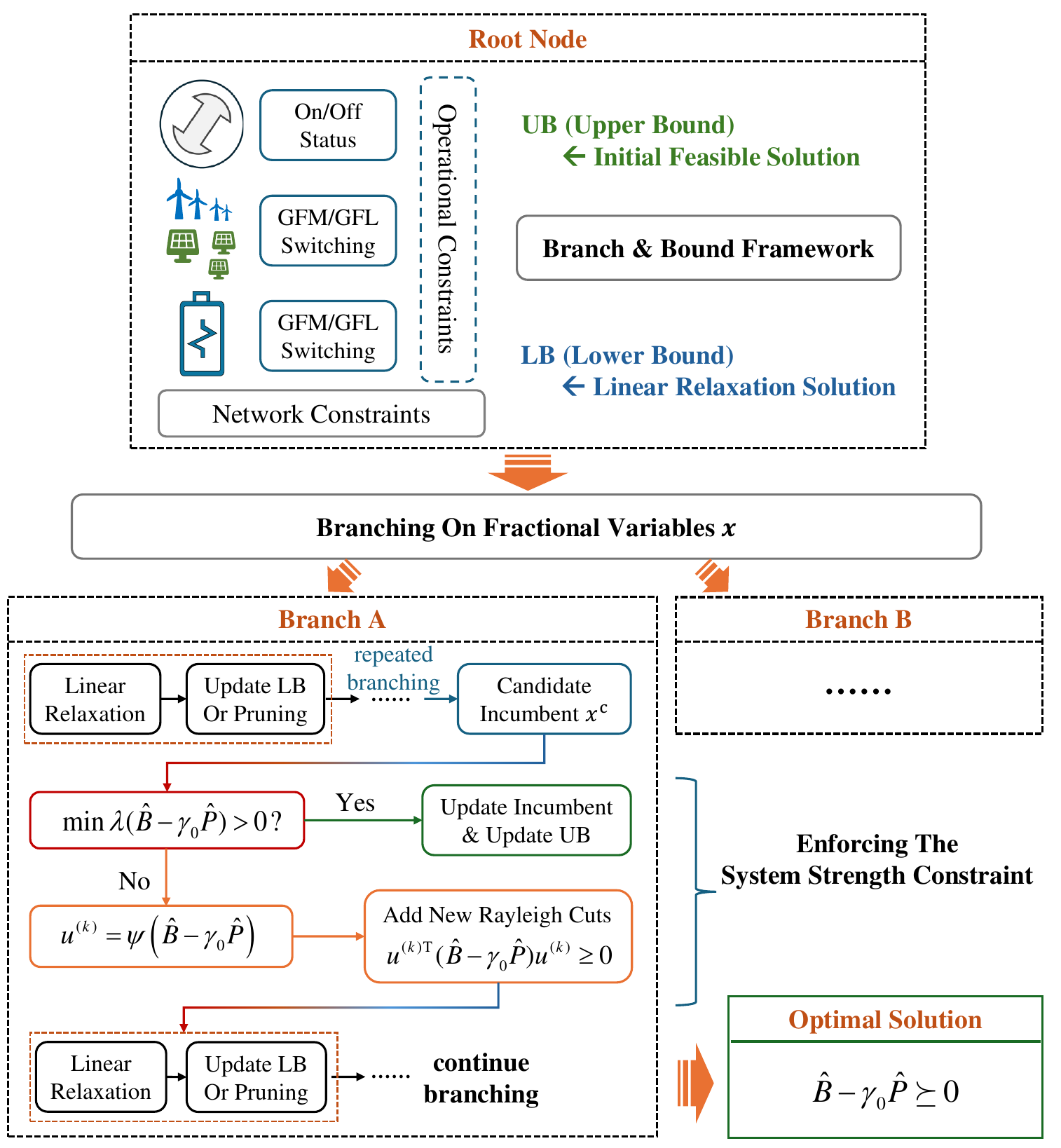}
    \caption{Rayleigh Cut method for solving the system strength-constrained scheduling problem. The branch-and-bound framework can be provided by commercial solvers such as Gurobi.}
    \label{fig:iterative-solution}
\end{figure}

\figref{fig:iterative-solution} illustrates the Rayleigh Cut method for iteratively solving the system strength-constrained scheduling problem within the branch-and-bound framework for MILP problems, which can be provided by commercial solvers such as Gurobi.
The procedure starts with an initial MILP problem with no Rayleigh Cuts, i.e., $\mathcal{U}=\emptyset$.
Every time when we obtain a candidate incumbent solution, we check whether it satisfies the system strength constraint \eqref{eq:final-lmi}.
If not, we compute the minimum eigenvector $u$ of the matrix $\hat{B}-\gamma_0\hat{P}$ based on this candidate incumbent solution and add the corresponding Rayleigh Cut constraints.
This procedure is repeated until a solution satisfying \eqref{eq:final-lmi} is found within the accepted MIP gap tolerance.
Here, we omit the details of the branch-and-bound framework and focus on the specific procedure for adding Rayleigh Cuts when a candidate incumbent solution is obtained, since the branch-and-bound framework is standard and well-known in the optimization community.
In addition, we note that the solution obtained by this method is guaranteed to be optimal for the original MISDP problem \eqref{eq:complete-scheduling}, within the accepted MIP gap tolerance, which is similar to the standard cutting plane method for solving MILP problems.

\section{Case Studies\label{sec:cases}}

\subsection{Basic Information}

The proposed strength-constrained scheduling model and solution method are evaluated on a modified IEEE 118-bus system.
The corresponding one-line diagram is shown in \figref{fig:network}.
The system comprises 118 buses, 186 transmission lines, and 50 power sources.
Among these sources, 32 are wind farms, each with a capacity of 150 MW; 8 are photovoltaic (PV) plants, also rated at 150 MW each; 7 are energy storage systems with a power capacity of 150 MW and an energy capacity of 300 MWh; and 3 are thermal generators with individual capacities of 667 MW, 300 MW, and 300 MW.
The load demand profile and renewable generation forecast for a 24-hour scheduling horizon are illustrated in \figref{fig:total-renewable}, where only the summation among all buses is presented for clarity.
The minimum required gOSCR level is set to $\gamma_0=2.0$ throughout the scheduling horizon.
\jxrevise{The power headroom requirement for GFM-mode IBRs is set to 10\% of the rated power capacity.
    \footnote{This setting is consistent with the guidance from the Australian Energy Market Operator (AEMO) \cite{AEMO2025GFMAccess} and the industry study practice of the Western Electricity Coordinating Council (WECC) \cite{WECC2023GFMInverter}.}
    The SOC reserve requirement for GFM energy storage systems is set to 8\% of the rated energy capacity.
    \footnote{This value is consistent with the operational practice reported for the Hornsdale Power Reserve (HPR) project in South Australia \cite{Aurecon2018HPR}.
        HPR is a Neoen-owned battery energy storage project, supplied by Tesla, with knowledge-sharing and technical reporting supported by the Australian Renewable Energy Agency (ARENA) and independent assessments by Aurecon.}}

\begin{figure}[thbp]
    \centering
    \includegraphics[width=0.99\linewidth]{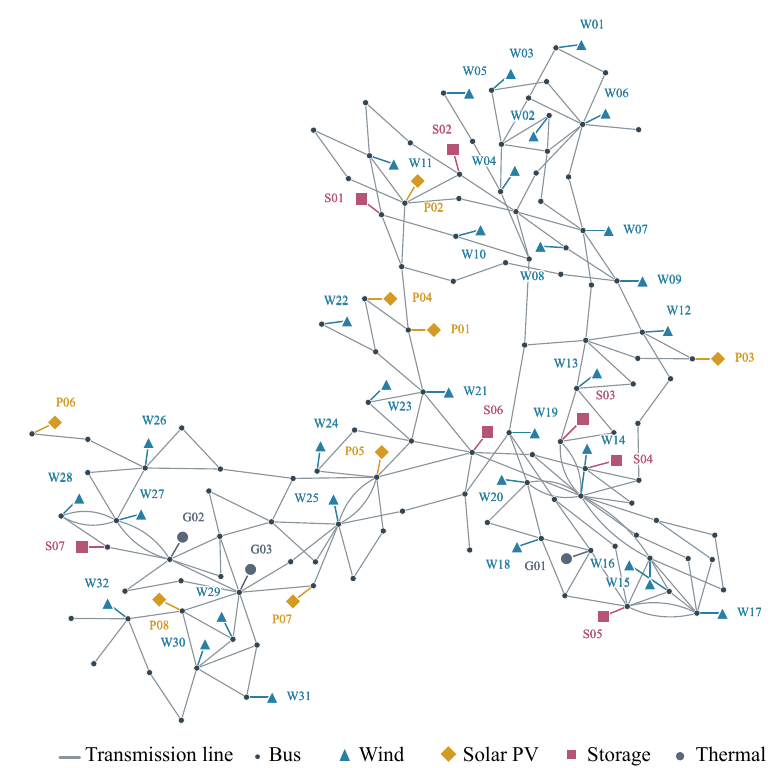}
    \caption{Network structure of the modified IEEE-118 bus system used in this paper. Black dots denote buses, gray lines denote transmission lines, blue triangles denote wind power plants, yellow diamonds denote photovoltaic plants, purple squares denote energy storage systems, and gray disks denote thermal power plants.}
    \label{fig:network}
\end{figure}

\begin{figure}[thbp]
    \centering
    \includegraphics[width=0.9\linewidth]{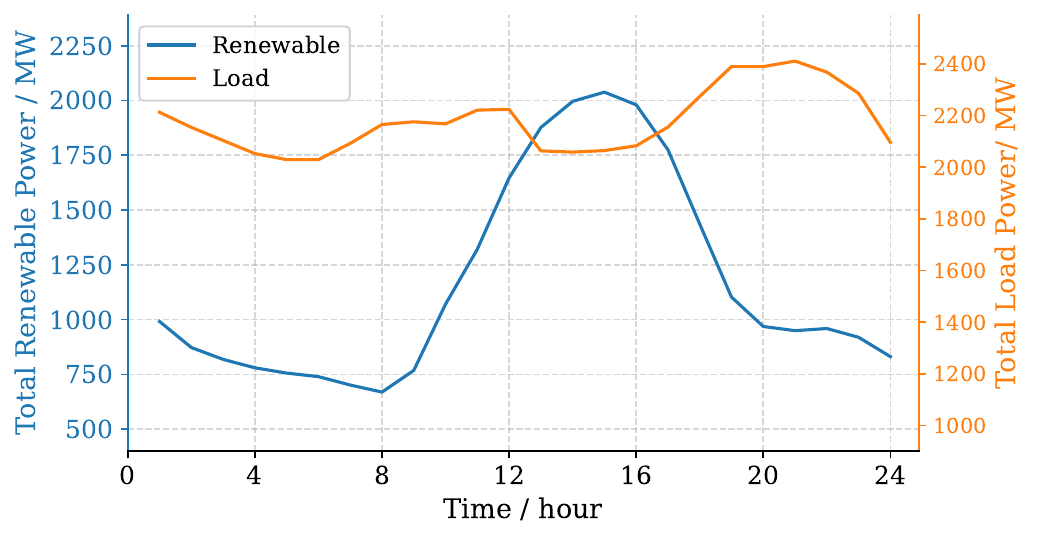}
    \caption{Renewable generation forecast and load demand at all time horizons. The blue line indicates the total renewable generation forecast, while the orange line indicates the total load demand.}
    \label{fig:total-renewable}
\end{figure}

The MILP problems are solved using Gurobi 12.0.
All the codes are run on a computer with an AMD Ryzen 9 9950X CPU and a 256 GB RAM.

\subsection{Model Performance}

The solving process of the proposed method is illustrated in \figref{fig:iteration-process}.
The blue curve represents the minimum gOSCR level across all time horizons, while the orange curve denotes the optimality gap.
The black dashed line indicates the system strength requirement $\gamma_0 = 2.0$.
\jxrevise{At each iteration, the proposed method evaluates the gOSCR levels over all time horizons based on the current scheduling solution.
    If the system strength requirement is violated, Rayleigh Cut constraints are added to remove the corresponding infeasible operating points and progressively enforce the required threshold.
    As shown in \figref{fig:iteration-process}, the MIP gap decreases rapidly to a small value in the early stage.
    The subsequent iterations mainly focus on adding Rayleigh Cuts to eliminate scheduling solutions that violate the system strength constraint, thereby automatically adjusting the scheduling result until all system strength requirements are satisfied.
    The final solution achieves a minimum gOSCR level of $2.0$ with an optimality gap of $0.17\%$.
    The total solving time is 32.96 seconds, involving 39 iterations and 249 added Rayleigh Cut constraints.}

\begin{figure}[thbp]
    \centering
    \includegraphics[width=0.95\linewidth]{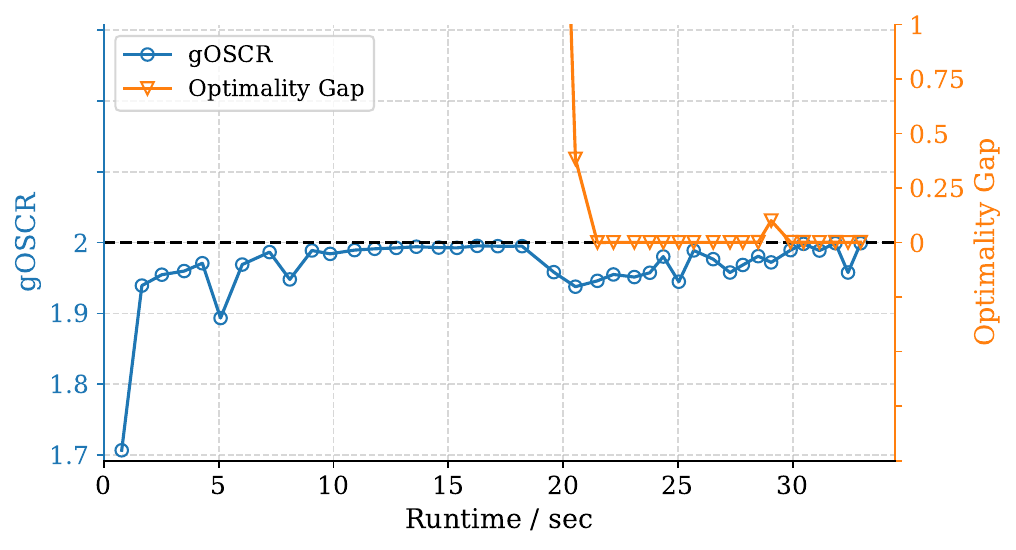}
    \caption{Solving process of the proposed method. The blue line indicates the gOSCR values, while the orange line indicates the optimality gap. The black dashed line indicates the system strength requirement $\gamma_0$.}
    \label{fig:iteration-process}
\end{figure}

\jxrevise{To further validate the reformulated system strength constraint \eqref{eq:final-lmi}, we compare its feasibility with that of the original constraint \eqref{eq:gOSCR-constraint} using 1,000 randomly generated instances with matrix dimensions ranging from 5 to 1,000.
    In \figref{fig:quadrant-consistency}, each blue dot represents one instance.
    The horizontal coordinate is the difference between \(\mathrm{gOSCR}\) and the required threshold \(\gamma_0\), while the vertical coordinate is the minimum eigenvalue of \(B-\gamma_0P\).
    Since the two constraints are theoretically equivalent, all points are expected to lie in the first or third quadrant, indicating identical feasibility outcomes.
    The numerical results confirm this expectation: all 1,000 instances lie in these two quadrants, providing numerical support for the equivalence between \eqref{eq:final-lmi} and \eqref{eq:gOSCR-constraint}.}

\begin{figure}[htbp]
    \centering
    \includegraphics[width=0.8\linewidth]{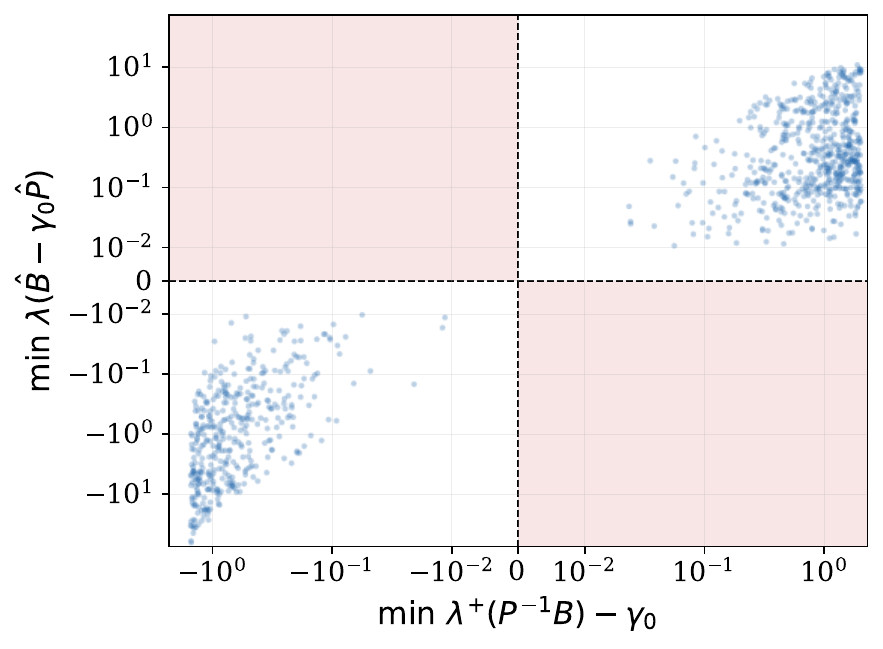}
    \caption{Numerical validation of the equivalence between the original system strength constraint \eqref{eq:gOSCR-constraint} and its reformulation \eqref{eq:final-lmi} using 1,000 randomly generated instances. Each blue dot represents one instance; its horizontal coordinate is the difference between \(\mathrm{gOSCR}\) and the required threshold \(\gamma_0\), and its vertical coordinate is the minimum eigenvalue of \(B-\gamma_0P\).}
    \label{fig:quadrant-consistency}
\end{figure}

\subsection{Interactions Between IBR Operating Behaviors and System Strength Requirement}

\figref{fig:gSCR-results} presents the system strength assessment results across all time horizons obtained by the proposed scheduling model.
The gOSCR levels remain above the required threshold of $2.0$ for all time periods, indicating that the proposed model effectively enforces the system strength requirement during the scheduling stage.
It can also be observed that system strength tends to decrease during periods with higher renewable generation.
This phenomenon can be explained using Theorem~\ref{thm:1}, which shows that the gOSCR level is strongly influenced by the power injections from GFL-mode renewable energy sources and energy storage systems.
\jxrevise{Specifically, the gOSCR level decreases with increasing power injections from these sources, as characterized by
    \begin{equation}
        \label{eq:power-sensitivity}
        \frac{\dif \gamma}{\dif p_i} = - \frac{\gamma \xi_i^2}{\xi^{\T} P \xi} \leq 0,
    \end{equation}
    where $p_i$ is the corresponding power injection and $\xi$ is the right eigenvector associated with $\gamma$. The detailed derivation is provided in Appendix~\ref{app:power-sensitivity}.}
Therefore, under high IBR penetration, system strength may become a critical time-varying bottleneck constraint, especially during periods with high renewable generation.
In \figref{fig:gSCR-results}, the system strength level is relatively low during periods 11-18 and reaches the required threshold in some periods, indicating that the system strength constraint becomes active or nearly active.
During these critical periods, the proposed scheduling model maintains the required system strength by coordinating multiple operational measures, including scheduling part of the renewable energy resources and energy storage systems in GFM mode, increasing the charging power of energy storage systems, and reducing their discharging power, which will be discussed in more detail in the following paragraphs.

\begin{figure}[htbp]
    \centering
    \includegraphics[width=0.87\linewidth]{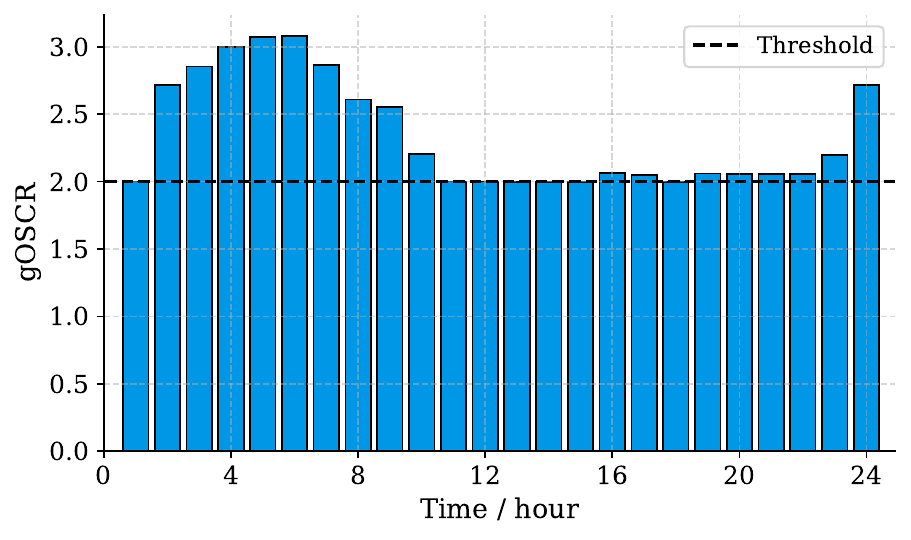}
    \caption{System strength assessment at all time horizons after scheduling using the proposed model. The black dashed line indicates the minimum required gOSCR level.}
    \label{fig:gSCR-results}
\end{figure}

For the energy storage systems, the state-of-charge (SoC) balance constraints require alternating charging and discharging across different periods.
However, the scheduling results indicate a clear tendency toward increased charging (i.e., negative power injection) during periods with higher renewable generation, as illustrated in \figref{fig:total-energy-storage}.
This behavior can be explained from two complementary perspectives.
From the system strength perspective, increased GFL-mode renewable generation tends to reduce the gOSCR level.
In contrast, charging of energy storage systems contributes negatively to power injection and, according to \eqref{eq:power-sensitivity}, increases the minimum eigenvalue of $B - \gamma_0 P$, thereby enhancing the gOSCR level.
Consequently, charging during high-renewable periods helps maintain system strength above the required threshold.
From the power balance perspective, energy storage systems absorb surplus renewable generation during high-output periods and discharge during low-renewable periods.
This operational pattern smooths the net load profile and facilitates system-wide power balance.
Therefore, it can be concluded that the power balance constraint and the system strength constraint have a synergistic effect on the scheduling of energy storage systems, resulting in increased charging during periods of high renewable generation and increased discharging during periods of low renewable generation.

\begin{figure}[htbp]
    \centering
    \includegraphics[width=0.87\linewidth]{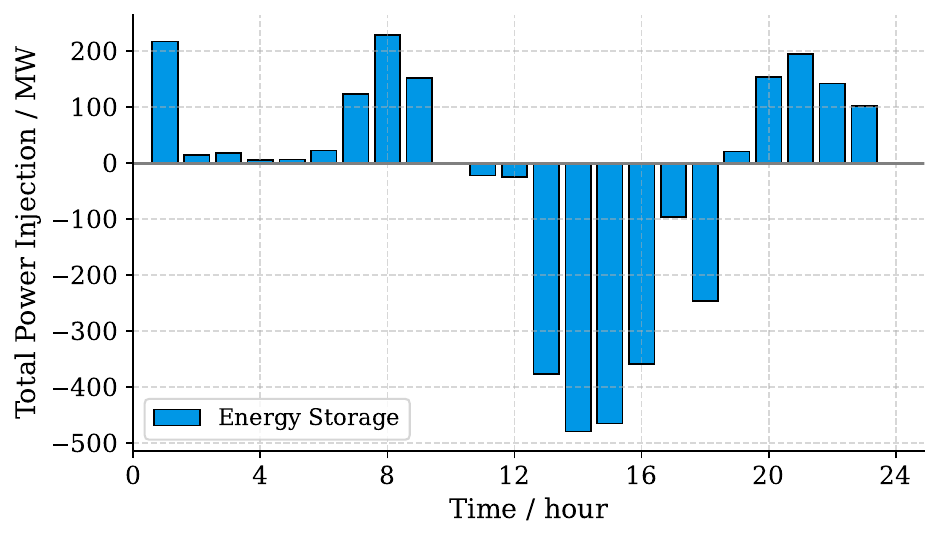}
    \caption{Total scheduled output of energy storage systems at all time horizons. Positive values indicate discharging, while negative values indicate charging.}
    \label{fig:total-energy-storage}
\end{figure}

In addition, switching to GFM mode significantly enhances system strength.
\figref{fig:one-energy-storage} presents the scheduling results of a representative energy storage system, where the green shaded regions denote the periods during which the unit operates in GFM mode.
During these intervals, the energy storage system must reserve a portion of its power and energy capacity to provide GFM services.
Although this reservation incurs an opportunity cost, it contributes to strengthening the system.
This trade-off is clearly reflected in the scheduling outcomes.
The energy storage system is more likely to operate in GFM mode during periods characterized by higher renewable generation and lower gOSCR levels, where additional system strength support is most needed.
This behavior further demonstrates that the proposed scheduling model effectively balances economic efficiency and system strength requirements.
To provide a quantitative illustration of the trade-off between system strength requirements and operational economy, we further conduct a sensitivity study by increasing the required system strength threshold by one unit while keeping all other settings unchanged.
The total operating cost increases by 207,541 CNY, approximately 30,000 USD, after resolving the scheduling problem.
As analyzed above, this cost increase is mainly caused by the need to schedule more IBRs in GFM mode to provide additional system strength support, which leads to higher opportunity costs.

\begin{figure}[htbp]
    \centering
    \includegraphics[width=0.95\linewidth]{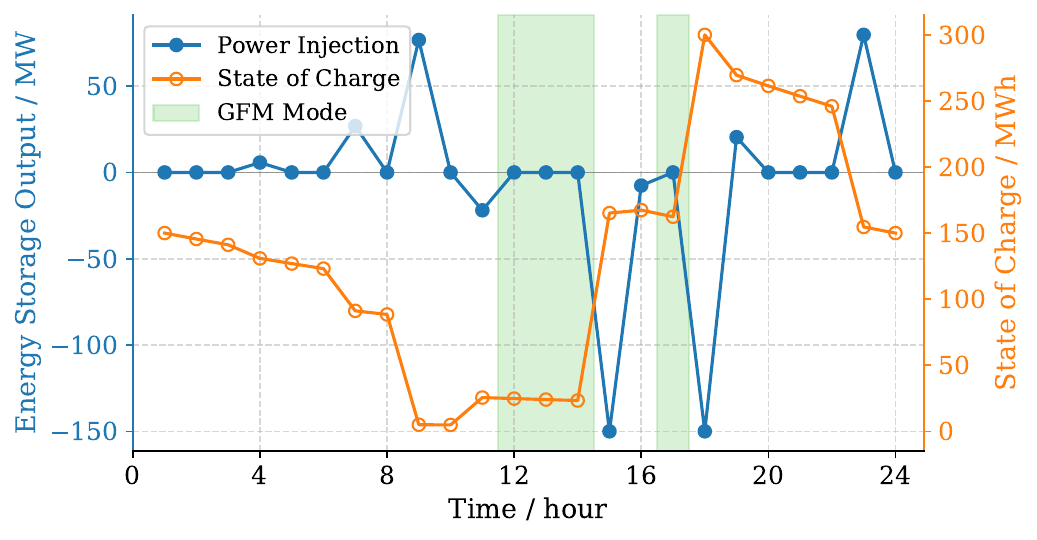}
    \caption{Scheduled output of one representative energy storage system. The blue line indicates the scheduled output, where positive values indicate discharging and negative values indicate charging. The orange line indicates the SoC of the energy storage system. The green area indicates the periods when the energy storage system operates in GFM mode.}
    \label{fig:one-energy-storage}
\end{figure}

For renewable resources such as wind farms and PV plants, a similar pattern can be observed: they are more likely to operate in GFM mode during periods characterized by higher renewable generation and lower gOSCR levels.
Seven renewable generators are scheduled to operate in GFM mode in our case studies, as illustrated in \figref{fig:all-renewables}, which shows the scheduled renewable generation and the reserved power for GFM operation at all time horizons.
This scheduling pattern yields two primary benefits.
First, operating in GFM mode during high-renewable periods enhances system strength and helps maintain the gOSCR level above the required threshold, as discussed previously.
Second, sufficient renewable generation provides the necessary headroom for GFM operation.
Moreover, the reserved power maintained for GFM services reduces the peak renewable output to some extent, which contributes to facilitating system-wide power balance.

\begin{figure}[htbp]
    \centering
    \includegraphics[width=0.99\linewidth]{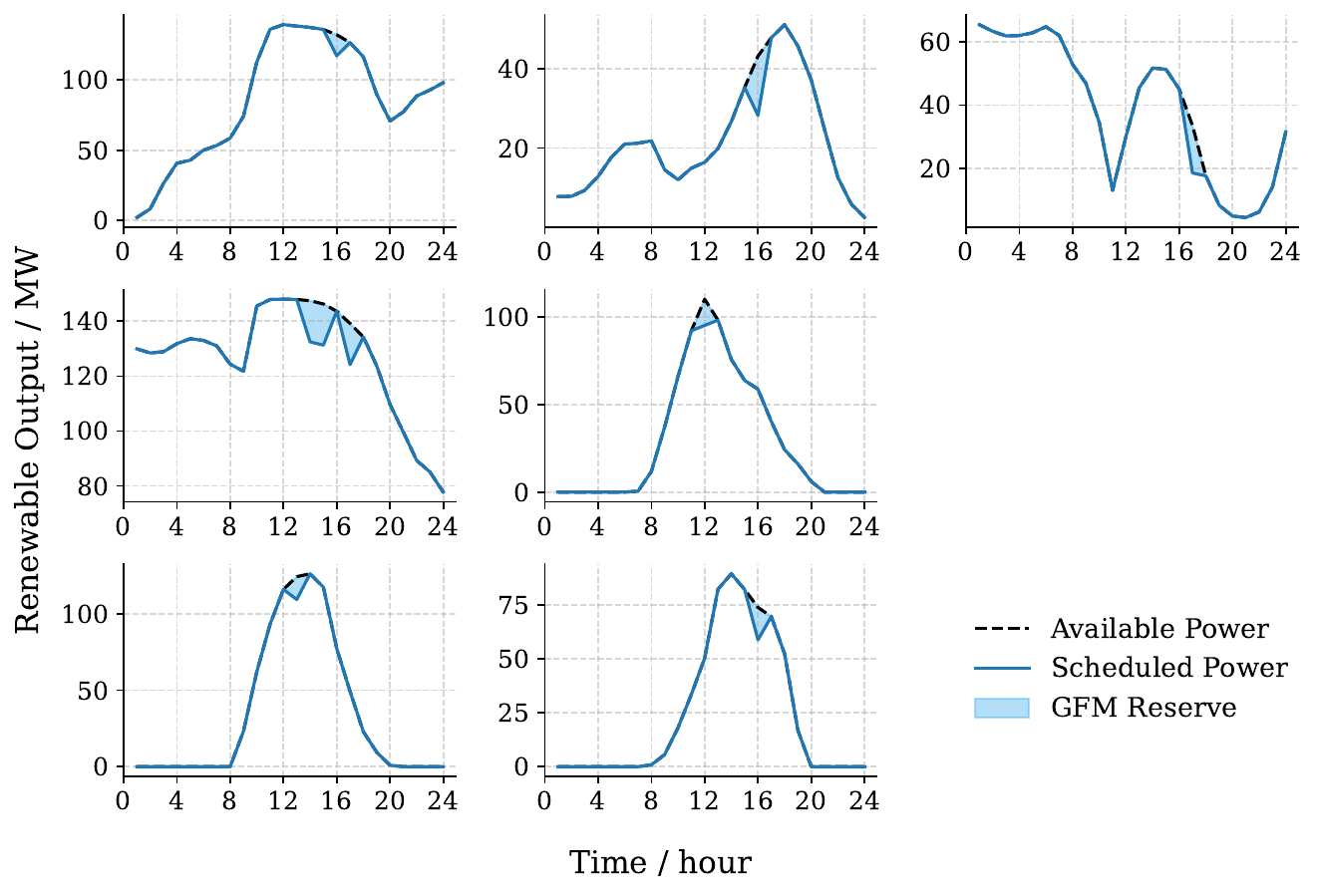}
    \caption{Scheduled renewable generations with GFM reserves at all time horizons. The black dashed line indicates the renewable generation forecast, the blue line indicates the scheduled renewable generation, while the light blue area indicates the power reserved for GFM operation.}
    \label{fig:all-renewables}
\end{figure}

We also conducted sensitivity studies on the power headroom parameter and the SOC reserve parameter.
First, while keeping the SOC reserve fixed at 8\%, we increased the required power headroom from 10\% to 15\%.
In this case, due to the higher headroom requirement, the total number of time periods in which renewable energy resources operate in GFM mode decreases from 9 to 3.
Meanwhile, to satisfy the system strength requirement, the total number of time periods in which energy storage systems operate in GFM mode increases from 12 to 24.
We then kept the power headroom requirement fixed at 10\% and increased the SOC reserve requirement from 8\% to 13\%.
In this case, the total number of GFM operating periods for renewable energy resources increases from 9 to 10, while that for energy storage systems decreases from 12 to 7.
These trends are consistent with the general qualitative understanding of the scheduling behavior.
In practical applications, the headroom parameter is reported by device manufacturers to the system operator and serves as one of the boundary conditions of the proposed scheduling model, similar to the data from conventional thermal units.
Its specific impact on the optimization results is difficult to predict precisely in advance.
However, once such data are available, the proposed scheduling model can be solved accordingly.

\subsection{Effects of Unit Commitment on IBR Mode Switching}

Under the baseline case-study setting discussed above, all thermal generators remain online throughout the 24-hour scheduling horizon, with no unit commitment transitions.
This behavior is mainly driven by the need to maintain sufficient thermal generation online during periods of low renewable generation to ensure power balance.
Although some thermal generators could be shut down during periods of high renewable generation, thermal units typically incur substantial startup costs relative to the operating costs of renewable generation.
As a result, shutting them down and recommitting them later is often not economically attractive, making continuous operation the more economical choice over the scheduling horizon.
To further investigate the impact of unit commitment on the GFM/GFL mode switching of IBRs, we reduce the startup costs and re-solve the system strength-constrained scheduling problem.
The resulting commitment status is shown in \figref{fig:uc-low}, where solid circles denote the on status, while hollow circles denote the off status.
It can be observed that the two 300-MW thermal generators are shut down during periods of high renewable generation, allowing the system to accommodate more low-cost renewable generation for power balancing.

\begin{figure}[thbp]
    \centering
    \includegraphics[width=0.85\linewidth]{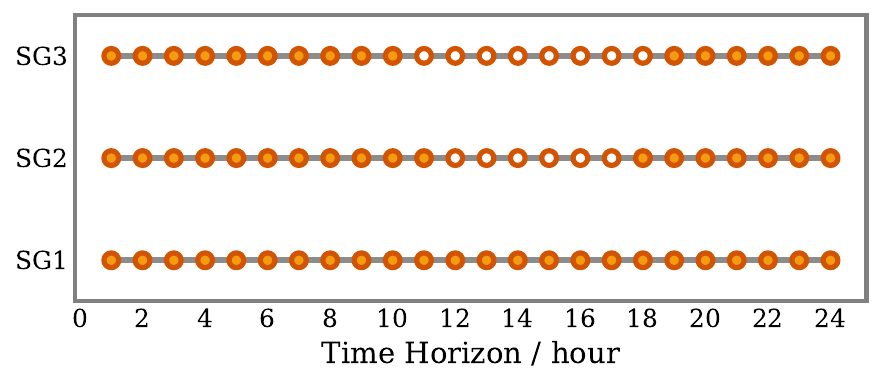}
    \caption{Thermal generator commitment status with reduced startup costs. Solid circles denote the on status, while hollow circles denote the off status. The horizontal axis represents time periods, and the vertical axis represents different synchronous generators.}
    \label{fig:uc-low}
\end{figure}

The shutdown of these thermal generators reduces the system strength support provided by synchronous generation, which is compensated by increased GFM operation of IBRs to maintain the gOSCR above the prescribed threshold.
In the baseline case, the numbers of renewable energy sources operating in GFM mode during hours 12-17 are only 1, 1, 1, 1, 3, and 2, respectively, as shown in \figref{fig:all-renewables}.
Under the modified case-study setting, five renewable energy sources are scheduled in GFM mode at each hour during hours 12-17, with the specific units selected for GFM operation varying over time among all renewable energy sources.
This result clearly illustrates the coordination between thermal unit commitment and GFM/GFL mode switching in maintaining adequate system strength.
This also indicates that satisfying the system strength requirement does not correspond to a unique operating pattern; multiple combinations of unit commitment and GFM/GFL mode switching can provide sufficient system strength, but with different operating costs.
The proposed scheduling framework identifies the economically optimal solution among these feasible alternatives under the given cost parameters, load forecasts, and other model inputs.

\subsection{Comparison with Direct MISDP Solution}

To further demonstrate the advantages of the proposed Rayleigh-cut solution method, we also solve the above case using a standard direct MISDP solution approach.
Specifically, the MISDP is handled by the branch-and-bound framework in YALMIP, with MOSEK employed to solve the SDP relaxation at each node.

The results demonstrate the significant computational advantage of the proposed Rayleigh-cut solution method.
For the direct MISDP solution, the first feasible solution is found after 45.7 minutes, and the MIP gap remains as high as $85.78\%$ after 2 hours, indicating that the solution is still far from optimal.
In contrast, the proposed method reaches a MIP gap of only $0.17\%$ within 32.96 seconds.
This substantial improvement is achieved because each Rayleigh cut removes a nonempty infeasible region from the search space, thereby progressively tightening the feasible domain and accelerating convergence toward the optimal solution.
\figref{fig:misdp-gOSCR} presents the system strength assessment results across all time horizons for the scheduling solution obtained by the direct MISDP method after 2 hours.
Compared with the results obtained by the proposed method in \figref{fig:gSCR-results}, the gOSCR levels from the direct MISDP solution are noticeably over-conservative, with the maximum gOSCR reaching 8.2, which is significantly higher than the required threshold of 2.0.
Such over-conservativeness leads to unnecessary economic costs, as the system is operated with excessive system strength support.
Specifically, more than 27 renewable energy sources are scheduled in GFM mode, and four unnecessary thermal unit commitment transitions are observed over the scheduling horizon, leading to additional GFM opportunity costs and thermal startup costs, respectively.

\begin{figure}[htbp]
    \centering
    \includegraphics[width=0.87\linewidth]{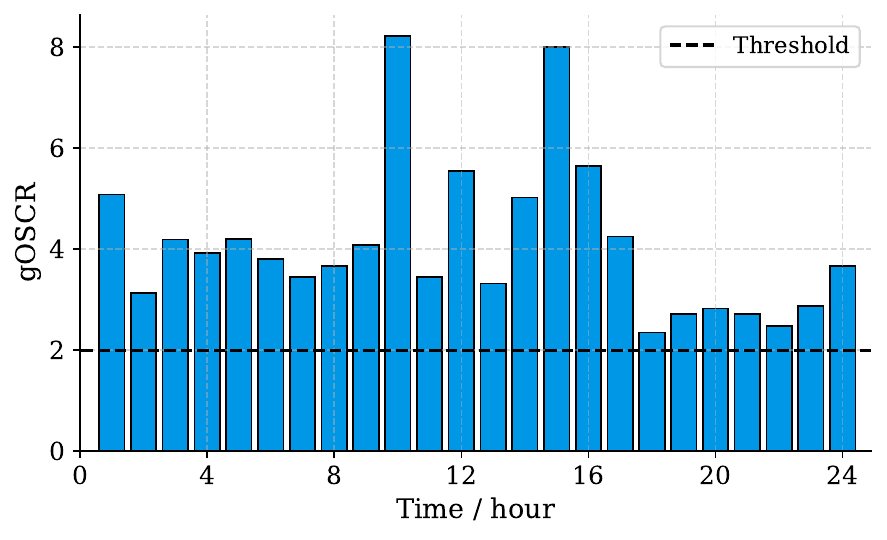}
    \caption{System strength assessment at all time horizons for the scheduling result obtained by the direct MISDP solution after 2 hours. The black dashed line indicates the minimum required gOSCR level.}
    \label{fig:misdp-gOSCR}
\end{figure}

\subsection{Application to a Real-World System}

To further validate the proposed scheduling model, we apply it to a transmission-level power system equivalently derived from a practical Jiangsu provincial power grid in China.
Geographically sensitive information has been removed from the system data.
The resulting test system consists of 194 buses and 238 transmission lines, and the data are obtained from a collaborative project with Jiangsu Power Grid.
Since this paper focuses on IBR-dominated power systems, the generation mix is further adjusted by increasing the penetration level of renewable energy resources.
The network topology is shown in \figref{fig:jiangsu-network}, where black dots denote buses, gray lines denote transmission lines, blue triangles denote wind power plants, yellow diamonds denote photovoltaic plants, purple squares denote energy storage systems, and gray disks denote thermal power plants.

\begin{figure}[htbp]
    \centering
    \includegraphics[width=0.99\linewidth]{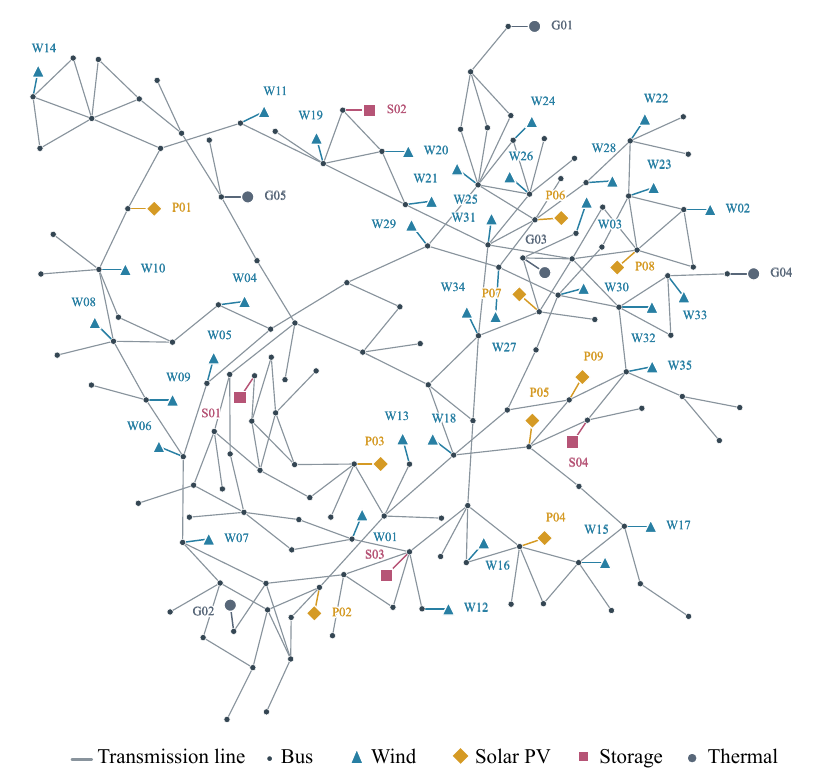}
    \caption{Network structure of the Jiangsu power system used in this paper. Black dots denote buses, gray lines denote transmission lines, blue triangles denote wind power plants, yellow diamonds denote photovoltaic plants, purple squares denote energy storage systems, and gray disks denote thermal power plants.}
    \label{fig:jiangsu-network}
\end{figure}

The proposed method solves the 24-hour system-strength-constrained scheduling problem for the Jiangsu system in 207.83 seconds.
\figref{fig:jiangsu-gOSCR} presents the system strength assessment results at each time horizon after scheduling.
It can be observed that the system strength levels remain above the required threshold for all time horizons.
This further verifies the accuracy of the proposed system strength constraints, which is theoretically guaranteed by the rigorous reformulation developed in this paper.
In particular, the proposed reformulation does not introduce additional approximation errors.
Similar to the modified IEEE 118-bus system, we also conduct a sensitivity study on the Jiangsu system by increasing the required system strength threshold by one unit while keeping all other settings unchanged.
The total operating cost increases by 1,041,309 CNY, approximately 150,000 USD.
This cost increase reflects the additional opportunity costs incurred by scheduling more IBRs in GFM mode to provide system strength support, which is consistent with the observations from the modified IEEE 118-bus system.

\begin{figure}[htbp]
    \centering
    \includegraphics[width=0.87\linewidth]{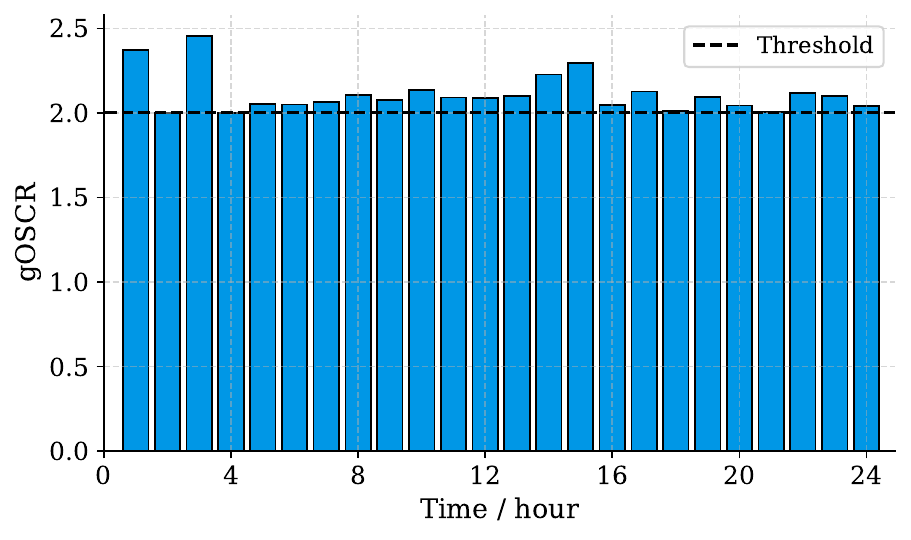}
    \caption{System strength assessment for the Jiangsu power system. The black dashed line indicates the minimum required gOSCR level.}
    \label{fig:jiangsu-gOSCR}
\end{figure}

\section{Conclusions\label{sec:conclusions}}

We have derived an equivalent LMI reformulation of the system strength constraint through a series of mathematical transformations, effectively addressing the non-convexity and dimension variation issues introduced by GFM/GFL mode switching of IBRs.
By integrating the LMI constraint with operational constraints, we establish a comprehensive system strength-constrained scheduling model, formulated as a MISDP problem, for IBR-dominated power systems with GFL/GFM mode switching considered.
To solve this problem, a Rayleigh Cut-based solution method is developed.
The proposed approach iteratively introduces cutting planes within the branch-and-bound framework of MILP problems, which is thus well compatible with commercial solvers, easily implementable, and optimality-guaranteed (within the MIP gap).
Case studies conducted on a modified IEEE 118-bus system demonstrate the effectiveness of the proposed method in maintaining system strength while optimizing scheduling decisions in IBR-dominated power systems.
The results reveal a clear synergistic interaction between power balance constraints and system strength constraints, offering valuable insights into the trade-offs between steady-state operational efficiency and dynamic stability requirements in power system scheduling.
Furthermore, the proposed framework provides a systematic tool for investigating operational patterns in IBR-dominated systems.
It captures how system strength constraints influence scheduling decisions and shape the operational behavior of IBRs, including their power outputs and GFL/GFM mode selections.

\appendix[]

\subsection{\texorpdfstring{Parameter Setting for the Big $M$ in \eqref{eq:x-A}}{}\label{app:big-M}}

Since \(M\) must be larger than the maximum of \(\abs{A_{ij,t}}\) for all \(i,j\in\mathcal{J}\) and \(t\in\mathcal{T}\), we derive a computable upper bound on $\abs{A_{ij,t}}$, which can be used to set the value of \(M\) in practice.
For notational simplicity, we omit the time index \(t\) in the following derivation.

Let $A=(B_{\mathcal{J}\mathcal{J}}^{\mathrm{sys}})^{-1}$ and denote the maximum and minimum singular values of $B_{\mathcal{J}\mathcal{J}}^{\mathrm{sys}}$ by $\sigma_{\max}$ and $\sigma_{\min}$, respectively.
It follows that
\begin{equation}
    \norm{A}_2=\norm{(B^{\mathrm{sys}}_{\mathcal{J}\mathcal{J}})^{-1}}_2=\frac{\sigma_{\max}}{\sigma_{\min}}\norm{B^{\mathrm{sys}}_{\mathcal{J}\mathcal{J}}}_2^{-1}.
\end{equation}
According to \eqref{eq:Bsys-x}, $B^{\mathrm{sys}}_{\mathcal{J}\mathcal{J}}=B^{\mathrm{pf}}_{\mathcal{J}\mathcal{J}}+D$, where $D$ is a diagonal matrix with non-negative entries.
Therefore,
\begin{equation}
    \norm{B^{\mathrm{sys}}_{\mathcal{J}\mathcal{J}}}_2\geq\frac{1}{\sqrt{n}}\norm{B^{\mathrm{pf}}_{\mathcal{J}\mathcal{J}}+D}_{\infty}\geq\frac{1}{\sqrt{n}} \norm{B^{\mathrm{pf}}_{\mathcal{J}\mathcal{J}}}_{\infty}.
\end{equation}
Combining these results gives
\begin{equation}
    \label{eq:app-norm}
    \norm{A}_{\infty}\leq \sqrt{n}\norm{A}_2\leq n\frac{\sigma_{\max}}{\sigma_{\min}}\norm{B^{\mathrm{pf}}_{\mathcal{J}\mathcal{J}}}_{\infty}^{-1}.
\end{equation}

Let \(\tau_{\max}\) and \(\tau_{\min}\) denote the maximum and minimum singular values of \(B_{\mathcal{J}\mathcal{J}}^{\mathrm{pf}}\), respectively, and let \(d_{\max}\) and \(d_{\min}\) denote the maximum and minimum diagonal entries of \(D\).
We further have
\begin{align}
    \label{eq:app-maxmin}
    \sigma_{\max} & =\max_{\norm{q}_2=1}(q^{\T}B^{\mathrm{pf}}_{\mathcal{J}\mathcal{J}}q+q^{\T}Dq)\leq \tau_{\max}+d_{\max},  \\
    \sigma_{\min} & =\min_{\norm{q}_2=1}(q^{\T}B^{\mathrm{pf}}_{\mathcal{J}\mathcal{J}}+q^{\T}Dq)\geq \tau_{\min} + d_{\min}.
\end{align}
Substituting \eqref{eq:app-maxmin} into \eqref{eq:app-norm} yields
\begin{equation}
    \norm{A}_{\infty}\leq n\frac{\tau_{\max}+d_{\max}}{\tau_{\min}+d_{\min}}\norm{B^{\mathrm{pf}}_{\mathcal{J}\mathcal{J}}}_{\infty}^{-1}=\kappa\norm{B^{\mathrm{pf}}_{\mathcal{J}\mathcal{J}}}_{\infty}^{-1},
\end{equation}
where $\kappa$ and $B^{\mathrm{pf}}_{\mathcal{JJ}}$ are independent of the decision variables.

Since
\(\abs{A_{ij}}\leq\norm{A}_{\infty}\), a valid choice is therefore
\begin{equation}
    M=
    \kappa
    \norm{B_{\mathcal{J}\mathcal{J}}^{\mathrm{pf}}}_{\infty}^{-1},
\end{equation}
which guarantees the validity of the big-M formulation.

\subsection{Illustrative Example}

To provide a concrete illustration of the reformulation developed in this paper, we consider the three-bus system shown in \figref{fig:example}.
The system comprises one synchronous generator connected to Bus 3 and two IBRs connected to Buses 1 and 2, respectively.
All line reactances are set to 0.05 p.u.
Using this simple system, we demonstrate how the original system strength constraint \eqref{eq:gOSCR-constraint}, which contains non-explicit expressions and decision variables of varying dimensions, can be equivalently transformed into an LMI constraint \eqref{eq:final-lmi} that is more tractable for optimization.

\begin{figure}[thbp]
    \centering
    \includegraphics[width=0.55\linewidth]{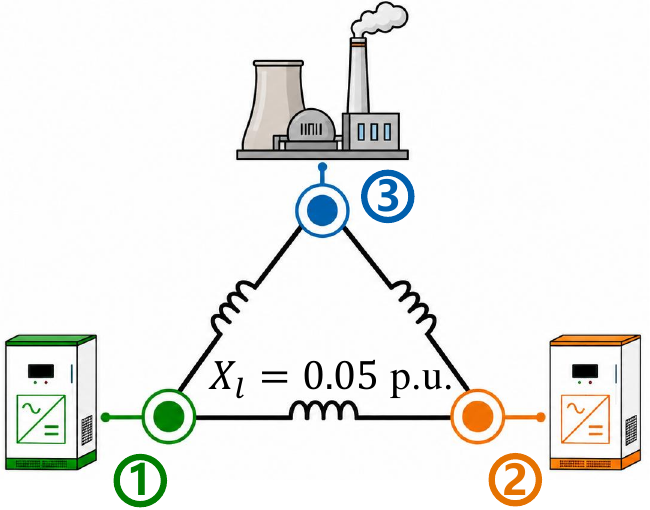}
    \caption{A three-bus system comprising one synchronous generator and two IBRs is considered for the illustrative example. The generator is connected to Bus 3, while the two IBRs are connected to Buses 1 and 2, respectively. The line reactances are given in per-unit, with all values set to 0.05 p.u.}
    \label{fig:example}
\end{figure}

The network admittance matrix of this system is given by
\begin{equation}
    B^{\mathrm{pf}} =
    \begin{bmatrix}
        40  & -20 & -20 \\
        -20 & 40  & -20 \\
        -20 & -20 & 40
    \end{bmatrix}.
\end{equation}
Using \eqref{eq:Bsys-x}, the corresponding system admittance matrix \(B^{\mathrm{sys}}\) can be written as
\begin{equation}
    B^{\mathrm{sys}} =
    \begin{bmatrix}
        40 - b_1 x_{1,t} & -20              & -20              \\
        -20              & 40 - b_2 x_{2,t} & -20              \\
        -20              & -20              & 40 - b_3 x_{3,t}
    \end{bmatrix}.
\end{equation}
Based on Corollary~\ref{cor:1}, the reduced matrix \( \hat{B} \) is obtained as
\begin{equation}
    \hat{B}\! =\!
    \begin{bmatrix}
        30\!-\!b_1x_{1,t}\!-\!\dfrac{10b_3x_{3,t}}{40-b_3}
         &
        -30-\dfrac{10b_3x_{3,t}}{40-b_3}
        \\[7pt]
        -30-\dfrac{10b_3x_{3,t}}{40-b_3}
         &
        30\!-\!b_2x_{2,t}\!-\!\dfrac{10b_3x_{3,t}}{40-b_3}
    \end{bmatrix}.
\end{equation}
After introducing the auxiliary variables \(\rho_{i,t}\) as defined in \eqref{eq:linaerize-P}, the reduced matrix \(\hat{P}\) in \eqref{eq:P-x} can be expressed as
\begin{equation}
    \hat{P} =\begin{bmatrix}
        p_{1,t}- \rho_{1,t} & 0                   \\
        0                   & p_{2,t}- \rho_{2,t}
    \end{bmatrix}.
\end{equation}

Accordingly, the original system strength constraint can be reformulated as
\(
\hat{B}-\gamma_0\hat{P}\succeq 0,
\)
which is an LMI constraint with a fixed dimension.
Notably, the effects of both the GFM/GFL mode switching of IBRs and the commitment status of synchronous generators are retained exactly in this reformulation, without introducing any modeling approximation.

\subsection{Derivative of the gOSCR to GFL Power Injections\label{app:power-sensitivity}}

Suppose \(\gamma\) is the gOSCR and \(\xi\) is the corresponding eigenvector, then we have
\begin{equation}
    P^{-1}B\xi = \gamma \xi \implies B\xi = \gamma P\xi.
\end{equation}
Differentiating both sides with respect to \(p_i\), we get
\begin{equation}
    B\frac{\dif \xi}{\dif p_i} = \frac{\dif \gamma}{\dif p_i} P \xi + \gamma \frac{\dif P}{\dif p_i}\xi + \gamma P \frac{\dif \xi}{\dif p_i}.
\end{equation}
Since \(\xi^{\T}B=\gamma\xi^{\T}P\), we multiply both sides by \(\xi^{\T}\) to obtain
\begin{equation}
    \frac{\dif\gamma}{\dif p_i} = -\gamma\frac{\xi_i^2}{\xi^{\T}P\xi}.
\end{equation}
We notice that
\begin{equation}
    \gamma\xi^{\T}P\xi = \xi^{\T}B\xi \geq 0 \implies \frac{\dif\gamma}{\dif p_i} = -\gamma\frac{\xi_i^2}{\xi^{\T}P\xi}\leq 0.
\end{equation}

\bibliographystyle{IEEEtran}
\bibliography{src/refs.bib}

@article{jiaxin2025synchronous,
  author   = {Wang, Jiaxin and Zhang, Jiawei and Hou, Qingchun and Zhang, Ning},
  journal  = {IEEE Transactions on Power Systems},
  title    = {Synchronous Condenser Placement for Multiple HVDC Power Systems Considering Short-Circuit Ratio Requirements},
  year     = {2025},
  volume   = {40},
  number   = {1},
  pages    = {765-779},
  doi      = {10.1109/TPWRS.2024.3404116}
}

@article{chenxi2024generalized,
  author   = {Liu, Chenxi and Xin, Huanhai and Wu, Di and Gao, Huisheng and Yuan, Hui and Zhou, Yuhan},
  journal  = {IEEE Transactions on Power Systems},
  title    = {Generalized Operational Short-Circuit Ratio for Grid Strength Assessment in Power Systems With High Renewable Penetration},
  year     = {2024},
  volume   = {39},
  number   = {4},
  pages    = {5479-5494},
  doi      = {10.1109/TPWRS.2023.3340158}
}

@article{nikos2021definition,
  author   = {Hatziargyriou, Nikos and Milanovic, Jovica and Rahmann, Claudia and Ajjarapu, Venkataramana and Canizares, Claudio and Erlich, Istvan and Hill, David and Hiskens, Ian and Kamwa, Innocent and Pal, Bikash and Pourbeik, Pouyan and Sanchez-Gasca, Juan and Stankovic, Aleksandar and Van Cutsem, Thierry and Vittal, Vijay and Vournas, Costas},
  journal  = {IEEE Transactions on Power Systems},
  title    = {Definition and Classification of Power System Stability – Revisited \& Extended},
  year     = {2021},
  volume   = {36},
  number   = {4},
  pages    = {3271-3281},
  doi      = {10.1109/TPWRS.2020.3041774}
}

@article{ning2023data,
  author   = {Zhang, Ning and Jia, Hongyang and Hou, Qingchun and Zhang, Ziyang and Xia, Tian and Cai, Xiao and Wang, Jiaxin},
  journal  = {Proceedings of the IEEE},
  title    = {Data-Driven Security and Stability Rule in High Renewable Penetrated Power System Operation},
  year     = {2023},
  volume   = {111},
  number   = {7},
  pages    = {788-805},
  doi      = {10.1109/JPROC.2022.3192719}
}

@article{kehao2026quantitative,
  author   = {Zhuang, Kehao and Xin, Huanhai and Chen, Hangyu and Huang, Linbin},
  journal  = {IEEE Transactions on Sustainable Energy},
  title    = {Quantitative Parameter Conditions for Stability and Coupling in GFM–GFL Converter Hybrid Systems From a Small-Signal Synchronous Perspective},
  year     = {2026},
  volume   = {},
  number   = {},
  pages    = {1-13},
  doi      = {10.1109/TSTE.2026.3662898}
}

@article{yunjie2023power,
  author   = {Gu, Yunjie and Green, Timothy C.},
  journal  = {Proceedings of the IEEE},
  title    = {Power System Stability With a High Penetration of Inverter-Based Resources},
  year     = {2023},
  volume   = {111},
  number   = {7},
  pages    = {832-853},
  doi      = {10.1109/JPROC.2022.3179826}
}

@article{qingchun2020impact,
  author   = {Hou, Qingchun and Du, Ershun and Zhang, Ning and Kang, Chongqing},
  journal  = {IEEE Transactions on Power Systems},
  title    = {Impact of High Renewable Penetration on the Power System Operation Mode: A Data-Driven Approach},
  year     = {2020},
  volume   = {35},
  number   = {1},
  pages    = {731-741},
  doi      = {10.1109/TPWRS.2019.2929276}
}

@article{yitong2022revisiting,
  author   = {Li, Yitong and Gu, Yunjie and Green, Timothy C.},
  journal  = {IEEE Transactions on Power Systems},
  title    = {Revisiting Grid-Forming and Grid-Following Inverters: A Duality Theory},
  year     = {2022},
  volume   = {37},
  number   = {6},
  pages    = {4541-4554},
  doi      = {10.1109/TPWRS.2022.3151851}
}

@article{huanhai2025many,
  author   = {Xin, Huanhai and Liu, Chenxi and Chen, Xia and Wang, Yuxuan and Prieto-Araujo, Eduardo and Huang, Linbin},
  journal  = {IEEE Transactions on Power Systems},
  title    = {How Many Grid-Forming Converters Do We Need? A Perspective From Small Signal Stability and Power Grid Strength},
  year     = {2025},
  volume   = {40},
  number   = {1},
  pages    = {623-635},
  doi      = {10.1109/TPWRS.2024.3393877}
}

@article{guoxuan2025many,
  author   = {Cui, Guoxuan and Jia, Hongyang and Zhang, Ning and Teng, Fei},
  journal  = {iEnergy},
  title    = {How many grid-forming converters are needed? — A techno-economic perspective},
  year     = {2025},
  volume   = {4},
  number   = {2},
  pages    = {79-85},
  doi      = {10.23919/IEN.2025.0012}
}

@article{krishayya1997ieee,
  title   = {{IEEE} guide for planning {DC} links terminating at {AC} locations having low short-circuit capacities, part {I}: {AC}/{DC} system interaction phenomena},
  author  = {Krishayya, PCS and others},
  journal = {CIGRE, France},
  pages   = {1-216},
  year    = {1997}
}

@techreport{cigre2008b4,
  title       = {System with Multiple {DC} Infeed},
  author      = {Davies, B. and others},
  year        = {2008},
  institution = {Cigr{\'e}, WG, B4.41},
  number      = {ELT\_241\_6},
  url         = {https://www.e-cigre.org/publications/detail/364-systems-with-multiple-dc-infeed.html}
}

@article{fuyilong2024assessing,
  author   = {Ma, Fuyilong and Xin, Huanhai and Wu, Di and Liu, Yun and Chen, Xia},
  journal  = {IEEE Transactions on Power Delivery},
  title    = {Assessing Small-Signal Grid Strength of 100\% Inverter-Based Power Systems},
  year     = {2024},
  volume   = {39},
  number   = {5},
  pages    = {2784-2796},
  doi      = {10.1109/TPWRD.2024.3432582}
}

@article{yue2024impedance,
  author   = {Zhu, Yue and Green, Timothy C. and Zhou, Xiaoyao and Li, Yitong and Kong, Dechao and Gu, Yunjie},
  journal  = {IEEE Transactions on Power Systems},
  title    = {Impedance Margin Ratio: A New Metric for Small-Signal System Strength},
  year     = {2024},
  volume   = {39},
  number   = {6},
  pages    = {7291-7303},
  doi      = {10.1109/TPWRS.2024.3371231}
}

@article{callum2024grid,
  author   = {Henderson, Callum and Egea-Alvarez, Agusti and Kneuppel, Thyge and Yang, Guangya and Xu, Lie},
  journal  = {IEEE Transactions on Power Delivery},
  title    = {Grid Strength Impedance Metric: An Alternative to SCR for Evaluating System Strength in Converter Dominated Systems},
  year     = {2024},
  volume   = {39},
  number   = {1},
  pages    = {386-396},
  doi      = {10.1109/TPWRD.2022.3233455}
}

@article{yunjie2021impedance,
  author   = {Gu, Yunjie and Li, Yitong and Zhu, Yue and Green, Timothy C.},
  journal  = {IEEE Transactions on Power Systems},
  title    = {Impedance-Based Whole-System Modeling for a Composite Grid via Embedding of Frame Dynamics},
  year     = {2021},
  volume   = {36},
  number   = {1},
  pages    = {336-345},
  doi      = {10.1109/TPWRS.2020.3004377}
}

@article{zhongda2023voltage,
  author   = {Chu, Zhongda and Teng, Fei},
  journal  = {IEEE Transactions on Power Systems},
  title    = {Voltage Stability Constrained Unit Commitment in Power Systems With High Penetration of Inverter-Based Generators},
  year     = {2023},
  volume   = {38},
  number   = {2},
  pages    = {1572-1582},
  doi      = {10.1109/TPWRS.2022.3179563}
}

@article{guoxuan2025control,
  author   = {Cui, Guoxuan and Chu, Zhongda and Teng, Fei},
  journal  = {IEEE Transactions on Power Systems},
  title    = {Control-Mode as a Grid Service in Software-Defined Power Grids: GFL vs GFM},
  year     = {2025},
  volume   = {40},
  number   = {1},
  pages    = {314-326},
  doi      = {10.1109/TPWRS.2024.3404339}
}

@article{yongkyu2025strength,
  author   = {Kim, Yong-Kyu and Lee, Sang-Ho and Lee, Gyu-Sub},
  journal  = {iEnergy},
  title    = {Strength-constrained unit commitment in IBR dominant power systems},
  year     = {2025},
  volume   = {4},
  number   = {2},
  pages    = {121-131},
  doi      = {10.23919/IEN.2025.0011}
}

@article{yongkyu2022evaluation,
  author   = {Kim, Yong-Kyu and Lee, Gyu-Sub and Yoon, Jong-Su and Moon, Seung-Il},
  journal  = {IEEE Transactions on Sustainable Energy},
  title    = {Evaluation for Maximum Allowable Capacity of Renewable Energy Source Considering AC System Strength Measures},
  year     = {2022},
  volume   = {13},
  number   = {2},
  pages    = {1123-1134},
  doi      = {10.1109/TSTE.2022.3152349}
}

@article{yuanhui2025placing,
  author   = {Yuan, Hui and Hao, Yi and Xin, Huanhai and Huang, Linbin and Zhou, Yuhan and Wang, Xiaofei and Chen, Chunmeng and Lu, Guoqiang and Qu, Linan and Wu, Di},
  journal  = {IEEE Transactions on Industry Applications},
  title    = {Placing Storage Energies for Enhancing Small-Signal Stability of Converter-Based-Renewable Systems},
  year     = {2025},
  volume   = {61},
  number   = {4},
  pages    = {5684-5698},
  doi      = {10.1109/TIA.2025.3546196}
}

@article{liwei2023hierarchical,
  author   = {Zhou, Liwei and Preindl, Matthias},
  journal  = {IEEE Transactions on Sustainable Energy},
  title    = {Hierarchical Software-Defined Control Architecture With MPC-Based Power Module to Interface Renewable Sources and Motor Drives},
  year     = {2023},
  volume   = {14},
  number   = {1},
  pages    = {83-96},
  doi      = {10.1109/TSTE.2022.3202957}
}

@article{huazhao2025novel,
  author   = {Ding, Huazhao and Kar, Rabi and Miao, Zhixin and Fan, Lingling},
  journal  = {IEEE Transactions on Sustainable Energy},
  title    = {A Novel Design for Switchable Grid-Following and Grid-Forming Control},
  year     = {2025},
  volume   = {16},
  number   = {2},
  pages    = {1301-1314},
  doi      = {10.1109/TSTE.2024.3520989}
}

@article{mahdi2025adaptive,
  author   = {Heidari, Mahdi and Ding, Lei and Kheshti, Mostafa and Zhao, Xiaowei and Terzija, Vladimir},
  journal  = {IEEE Transactions on Sustainable Energy},
  title    = {Adaptive Inertial Control for Wind Turbine Generators in Fast Frequency Response Based on the Power Reduction Period Assessment},
  year     = {2025},
  volume   = {16},
  number   = {1},
  pages    = {377-391},
  doi      = {10.1109/TSTE.2024.3459729}
}

@article{linbin2026system,
  author   = {Huang, Linbin and He, Xiuqiang and Xin, Huanhai and Li, Zhiyi and Ju, Ping and Ma, Fuyilong and Wang, Kang and Dörfler, Florian},
  journal  = {IEEE Power and Energy Magazine},
  title    = {System Strength in Power Electronics-Dominated Power Systems: An Enabler to Stability and Control},
  year     = {2026},
  volume   = {24},
  number   = {1},
  pages    = {49-66},
  doi      = {10.1109/MPE.2025.3599131}
}

@article{chu2023stabilityconstrainedoptimizationhigh,
  author   = {Chu, Zhongda and Teng, Fei},
  journal  = {IEEE Transactions on Sustainable Energy},
  title    = {A Unified Framework for Multi-Stability Constrained Optimization in IBR-Dominated Power Systems},
  year     = {2026},
  volume   = {},
  number   = {},
  pages    = {1-14},
  doi      = {10.1109/TSTE.2026.3692222}
}

@article{yonghong2023security,
  author   = {Chen, Yonghong and Pan, Feng and Qiu, Feng and Xavier, Alinson S and Zheng, Tongxin and Marwali, Muhammad and Knueven, Bernard and Guan, Yongpei and Luh, Peter B. and Wu, Lei and Yan, Bing and Bragin, Mikhail A. and Zhong, Haiwang and Giacomoni, Anthony and Baldick, Ross and Gisin, Boris and Gu, Qun and Philbrick, Russ and Li, Fangxing},
  journal  = {IEEE Transactions on Power Systems},
  title    = {Security-Constrained Unit Commitment for Electricity Market: Modeling, Solution Methods, and Future Challenges},
  year     = {2023},
  volume   = {38},
  number   = {5},
  pages    = {4668-4681},
  doi      = {10.1109/TPWRS.2022.3213001}
}

@article{Kamwa2026future,
  author   = {Kamwa, Innocent and Badrzadeh, Babak},
  journal  = {IEEE Power and Energy Magazine},
  title    = {Future Power Systems: Dynamic Stability Must Lead the Way},
  year     = {2026},
  volume   = {24},
  number   = {1},
  pages    = {4-19},
  doi      = {10.1109/MPE.2025.3629230}
}

@article{niu2026seamless,
  author   = {Niu, Shaokun and Chen, Alian and Huang, Yaopeng and Zhang, Guanguan and Liu, Tong and Cheng, Cheng},
  journal  = {IEEE Transactions on Industrial Electronics},
  title    = {A Seamless Switching Strategy Between Andronov-Hopf Oscillator-Based Grid-Forming Mode and Grid-Following Mode for Grid-Connected Inverters},
  year     = {2026},
  volume   = {},
  number   = {},
  pages    = {1-12},
  doi      = {10.1109/TIE.2026.3692834}
}

@article{juelin2022explicit,
  author   = {Liu, Juelin and Yang, Zhifang and Zhao, Junbo and Yu, Juan and Tan, Bendong and Li, Wenyuan},
  journal  = {IEEE Transactions on Power Systems},
  title    = {Explicit Data-Driven Small-Signal Stability Constrained Optimal Power Flow},
  year     = {2022},
  volume   = {37},
  number   = {5},
  pages    = {3726-3737},
  doi      = {10.1109/TPWRS.2021.3135657}
}

@article{fu2013modeling,
  author   = {Fu, Yong and Li, Zuyi and Wu, Lei},
  journal  = {IEEE Transactions on Power Systems},
  title    = {Modeling and Solution of the Large-Scale Security-Constrained Unit Commitment},
  year     = {2013},
  volume   = {28},
  number   = {4},
  pages    = {3524-3533},
  doi      = {10.1109/TPWRS.2013.2272518}
}

@article{wang2024stability,
  author   = {Wang, Jun and Fan, Feilong and Song, Yue and Hou, Yunhe and Hill, David J.},
  journal  = {IEEE Transactions on Sustainable Energy},
  title    = {Stability Constrained Optimal Operation of Inverter-Dominant Microgrids: A Two Stage Robust Optimization Framework},
  year     = {2024},
  volume   = {15},
  number   = {3},
  pages    = {1900-1913},
  doi      = {10.1109/TSTE.2024.3387296}
}

@article{wu2024transient,
  author   = {Wu, Tao and Wang, Jianhui},
  journal  = {IEEE Transactions on Neural Networks and Learning Systems},
  title    = {Transient Stability-Constrained Unit Commitment Using Input Convex Neural Network},
  year     = {2024},
  volume   = {35},
  number   = {11},
  pages    = {16023-16035},
  doi      = {10.1109/TNNLS.2023.3291673}
}

@techreport{AEMO2025GFMAccess,
  author      = {{Australian Energy Market Operator}},
  title       = {{Grid-forming Technology Access Standards: Approach Paper}},
  institution = {{Australian Energy Market Operator}},
  year        = {2025}
}

@techreport{WECC2023GFMInverter,
  author      = {{Western Electricity Coordinating Council}},
  title       = {{Grid Forming Inverter Study Report}},
  institution = {{Western Electricity Coordinating Council}},
  year        = {2023}
}

@techreport{Aurecon2018HPR,
  author      = {{Aurecon}},
  title       = {{Hornsdale Power Reserve: Year 1 Technical and Market Impact Case Study}},
  institution = {{Aurecon}},
  year        = {2018}
}

\end{document}